\documentclass[runningheads]{llncs}

\usepackage{makeidx}  % allows for indexgeneration
\usepackage[bookmarksnumbered,unicode]{hyperref}
\usepackage{amsmath}
	\usepackage{amsfonts}
	\usepackage{amssymb}
\usepackage{mathtools} % extends amsmath package
\usepackage[inline,shortlabels]{enumitem} %ELFAR: For enumerations without left indent
\usepackage{booktabs} % For formal tables
\usepackage{tikz} % For tikz graphics files
	\usetikzlibrary{automata,positioning,shapes,fit}
	\usetikzlibrary{calc}
\usepackage{pgfplots} % For pgf plots
	\usepgfplotslibrary{units}  %'/tikz/bar shift auto'
	\usepgfplotslibrary{fillbetween} % fill area between curves
	\usetikzlibrary{patterns} % for fill patterns
	\pgfplotsset{compat=1.7}
	\pgfplotsset{compat=newest} 
	\pgfplotsset{plot coordinates/math parser=false} 
	\newlength\figureheight 
	\newlength\figurewidth 
\usepackage{tikz-cd}
\usepackage{wrapfig}
\usepackage{cutwin}
\usepackage[ruled,vlined,linesnumbered]{algorithm2e}
\usepackage{graphicx}
	\graphicspath{ {figs/} }
\usepackage{etoolbox} % For toggling comments
\usepackage{upgreek}
\usepackage{pifont} % For TODO list
\usepackage{relsize} %to resize math symbols
\usepackage{fancyhdr} % Only for drafts
\usepackage[yyyymmdd,hhmmss]{datetime}

\usepackage{soul} % for highlighting text
\usepackage[tight,footnotesize]{subfigure}
\usepackage{stackengine}
\usepackage{printlen}
\usepackage[all,cmtip]{xy}
\usepackage[ligature, inference]{semantic}
\usepackage{xifthen}

\newcommand{\tableref}		[1]{Table~\ref{#1}}
\newcommand{\figref}		[1]{Fig.~\ref{#1}}
\newcommand{\thmref}		[1]{Theorem~\ref{#1}}
\newcommand{\lemref}		[1]{Lemma~\ref{#1}}
\newcommand{\exref}			[1]{Example~\ref{#1}}
\newcommand{\algref}		[1]{Algorithm~\ref{#1}}
\newcommand{\secref}		[1]{Sec.~\ref{#1}}

\newcommand{\defref}		[1]{Definition~\ref{#1}}

\newcommand{\ie}{i.e., }
\newcommand{\eg}{e.g., }
\newcommand{\suchthat}{\mbox{s.t.}}

\newcommand	{\prob}	[1]	{\langle{}#1\rangle{}}
\newcommand	{\RN} 	[1]	{\textup{\uppercase\expandafter{\romannumeral#1}}}
\newcommand	{\V}	[1]	{\text{\texttt{#1}}}
\newcommand	{\PL}	[1]	{PL\textsubscript{\RN{#1}}}
\newcommand {\PLs}		{PL\textsubscript{$\Psto$}{}}

\newcommand	{\Move}			{\mathit{move}}
\newcommand	{\plfn}	[1] %
	{\ifx&#1&%
	 \gamma
	 \else
  	 \gamma\left(#1\right)
	 \fi}
\newcommand {\Eval} [1] {\mathit{Ev}\!\left(#1\right)}
\newcommand{\Eff}[1]{%
		\ifthenelse{\isempty{#1}}%
		{\mathit{Eff}}% else
		{\mathit{Eff}\!\!\left(#1\right)}}

\newcommand	{\Val}	[1]	{\eta\!\left(#1\right)}
\newcommand {\Proper}{\mathrm{Prop}}
\newcommand {\First}{\mathit{first}}
\newcommand {\Last}{\mathit{last}}

\newcommand {\Ts}{\mathrm{TS}}
\newcommand {\Hyp}{\mathrm{Hyp}}
\newcommand {\Uav}{\mathrm{uav}}
\newcommand {\Oas}{\mathrm{as}}
\newcommand {\Adv}{\mathrm{adv}}
\newcommand {\Mrd}{\mathrm{mrd}}
\newcommand {\Mwr}{\mathrm{mwr}}
\newcommand {\Syn}{\mathrm{syn}}
\newcommand {\Ana}{\mathrm{ana}}

\newcommand {\Write}{\mathit{write}}
\newcommand {\Read}{\mathit{read}}
\newcommand {\Activate}{\mathit{activate}}
\newcommand {\Reach}{\mathit{reach}}
\newcommand {\Hazard}{\mathit{hazard}}
\newcommand {\Locate}{\mathit{locate}}
\newcommand {\Target}{\mathit{target}}
\newcommand {\Safe}{\mathit{safe}}
\newcommand {\Probability}{\mathrm{Pr}}

\newcommand {\Psto} 	{\bigcirc}
\newcommand	{\Pone}		{\RN{1}}
\newcommand	{\Ptwo}		{\RN{2}}
\newcommand	{\GX}		{\mathcal{G}} % Game
\newcommand	{\M}		{\mathcal{M}} % Automata

\newcommand	{\ZP}		{\mathbb{N}_0} %{\mathbb{Z}_{\geq 0}}

\newcommand	{\pidx}		{\gamma} %{{\scalebox{0.7}{$\pluto$}}}
\newcommand	{\B}		{\mathcal{B}} %{\flat}
\newcommand	{\T}		{\mathcal{T}} %{\hbar}

\newcommand	{\bspace}	{\Omega}
\newcommand	{\PRISM}	{{PRISM-games}}

\newcommand {\always}{\mathsf{G}}%{\square}
\newcommand {\eventually}{\mathsf{F}}%{\lozenge}
\newcommand {\until}{\mathsf{U}}
\newcommand {\query}{\phi}
\newcommand {\llangle}{\langle\!\langle} % <<
\newcommand {\rrangle}{\rangle\!\rangle} % >>
\newcommand {\simulate}{\leadsto}

\newcommand {\Eq}{\!=\!}
\newcommand {\DAG}{\mathcal{G}_{\mathsf{D}}}
\newcommand {\SMG}{\mathcal{G}}
\newcommand {\HIG}{\mathcal{G}_{\mathsf{H}}}
\newcommand {\SUB}{\hat{\mathcal{G}}}

\newcommand{\Reveal}{\theta}
\newcommand{\NoAction}{\tau}
\newcommand{\Stop}{\mathit{stop}}
\newcommand{\PlusOne}{\raisebox{.2\height}{\scalebox{.8}{+}}1}
\newlength{\xywd}
\SetKwComment{Comment}{$\triangleright$\ }{}
\newcommand{\CC}[1]{\Comment*[f]{#1}}
\newcommand{\DD}[1]{\Comment*[r]{#1}}
\SetSideCommentRight
\SetCommentSty{text}
\definecolor{azure(colorwheel)}{rgb}{0.0, 0.5, 1.0}
\newcommand{\RememberNode}[2]{%% "rn": "remember node"
    \tikz[remember picture,baseline=(#1.center)]\node [inner sep=1] (#1) {$#2$};%
}
\newcount\ppnum
\newcommand\num[1]{%
        \ppnum=#1\relax
        \ifnum\ppnum<0
                $-$%
                \ppnum=-\ppnum
        \fi
        \let\pptemp\empty
        \loop\ifnum\ppnum>999
                \count255=\ppnum
                \divide\ppnum by1000
                \count255=\numexpr \count255 - 1000*\ppnum \relax
                \edef\pptemp{,\ifnum\count255<100 0\ifnum\count255<10 0\fi\fi
                             \the\count255 \pptemp}%
        \repeat
        \the\ppnum
        \pptemp
}
\begin{document}

%% MAIN BODY %%%%%%%%%%%%%%%%%%%%%%%%%%%%%%%%%%%%%%%%%%%%%%%%%%%%%%%%%%%%%%%%%%
\mainmatter % start of the contributions

%%% TITLE %%%%%%%%%%%%%%%%%%%%%%%%%%%%%%%%%%%%%%%%%%%%%%%%%%%%%%%%%%%%%%%%%%%%%
\title{%
	Security-Aware Synthesis Using\\%
	Delayed-Action Games%
	\thanks{%
		This work was supported by the NSF CNS-1652544 grant,
		as well as the ONR N00014-17-1-2012 and N00014-17-1-2504,
		and AFOSR FA9550-19-1-0169 awards.}%
	}
\titlerunning{Security-Aware Synthesis Using Delayed-Action Games}
	% abbreviated title (for running head) also used for the TOC unless
	% \toctitle is used
%
%%% AUTHORS %%%%%%%%%%%%%%%%%%%%%%%%%%%%%%%%%%%%%%%%%%%%%%%%%%%%%%%%%%%%%%%%%%%
\author{%
	Mahmoud Elfar%
		\orcidID{0000-0002-5579-1255}
	\and
	Yu Wang%
		\orcidID{0000-0002-0431-1039}
	\and
	Miroslav Pajic%
		\orcidID{0000-0002-5357-0117}
	}
\authorrunning{M. Elfar et al.}
% list of authors for the TOC (use if author list has to be modified)
%\tocauthor{Mahmoud Elfar, Yu Wang, and Miroslav Pajic}
\institute{%
	Duke University, Durham NC 27708, USA\\	
	\email{\{mahmoud.elfar,yu.wang094,miroslav.pajic\}@duke.edu}
}

%%% INPUT FILES %%%%%%%%%%%%%%%%%%%%%%%%%%%%%%%%%%%%%%%%%%%%%%%%%%%%%%%%%%%%%%%
\maketitle % typeset the title of the contribution

%%% ABSTRACT %%%%%%%%%%%%%%%%%%%%%%%%%%%%%%%%%%%%%%%%%%%%%%%%%%%%%%%%%%%%%%%%%%
%!TEX root = 00_CAV19_DAG_Cam.tex
%%%%%%%%%%%%%%%%%%%%%%%%%%%%%%%%%%%%%%%%%%%%%%%%%%%%%%%%%%%%%%%%%%%%%%%%%%%%%%%

\begin{abstract}
Stochastic multiplayer games (SMGs) have gained attention in the field of strategy synthesis for multi-agent reactive systems.
However, standard SMGs are limited to modeling systems
 where all agents have full knowledge of the state of the game.
%\ME{\TBD\ Add sentence to establish a niche, \eg why not POSMGs?} 
%A very limited number of studies considered studying games with some players suffering partial blindness to the state of the game, though such partial-information games could be more suited for many real systems.
%%%In this paper, we introduce hidden-information games (HIGs) --- a class of games where one coalition of players can probabilistically and intermittently gain full access to the current state.
%We provide the semantics of HIGs using private variables.
%%%Next,
In this paper, we introduce delayed-action games (DAGs) formalism that simulates 
 hidden-information games (HIGs) as SMGs, where hidden information is captured by delaying
 a player's actions.
 %adversary actions, % and their semantics.
%%and we further show how DAGs can be used to simulate HIGs. %simulation relation between HIGs and DAGs.
%By eliminating the usage of private variables,
The elimination of private variables enables the usage of SMG off-the-shelf model checkers to implement HIGs.
%we demonstrate how DAGs can be implemented using off-the-shelf SMG model checkers.
%%%Furthermore, %by enforcing a finite horizon, 
%%%we demonstrate how a DAG can be decomposed
%%% into a number of independent subgames.
%%%Since each subgame can be independently explored, parallel computation can be utilized to reduce the model checking time, while alleviating the state space explosion problem that SMGs are notorious for.
Furthermore, we demonstrate how a DAG can be decomposed into subgames that can be independently explored, utilizing parallel computation to reduce the model checking time, while alleviating the state space explosion problem that SMGs are notorious for.
In addition, we propose a DAG-based framework for strategy synthesis and analysis.
Finally, we demonstrate applicability of the DAG-based synthesis framework on a case study of a human-on-the-loop unmanned-aerial vehicle system 
%%%that may be under stealthy attack,
under stealthy attacks,
where the proposed framework is used to formally model, analyze and synthesize security-aware strategies for the system.
\end{abstract}

\renewcommand{\xrightarrow}[2][]{%
  \sbox{0}{$\scriptstyle#2$}%
  \xywd=\wd0%
  \!\!\!\xymatrix@C\dimexpr\xywd+0.6em\relax{{}\ar[r]^{\!\!#2}&{\!\!\!}\!}\!\!%
}
% No Inference for me
\renewcommand{\inference}[4][]{#2  & \implies &  #3 % 
	\ifthenelse{\isempty{#4}}%
		{}% else
		{ \; \mbox{s.t.} \; #4}}
% Custom spacings
\setlength{\textfloatsep}{1.0\baselineskip plus 0.2\baselineskip minus 0.5\baselineskip}
\setlength{\abovecaptionskip}{3pt plus 0pt minus 7pt}
\setlength{\abovedisplayskip}{5pt}
\setlength{\belowdisplayskip}{5pt}
\makeatletter
\renewcommand\subsubsection{\@startsection{subsubsection}{3}{\z@}%
	{-12\p@ \@plus -4\p@ \@minus -4\p@}%
	{-0.5em \@plus -0.22em \@minus -0.1em}%
	{\normalfont\normalsize\bfseries\boldmath}} \makeatother          

%%% INPUT FILES %%%%%%%%%%%%%%%%%%%%%%%%%%%%%%%%%%%%%%%%%%%%%%%%%%%%%%%%%%%%%%%

%!TEX root = 00_CAV_2019_paper_78.tex
%%%%%%%%%%%%%%%%%%%%%%%%%%%%%%%%%%%%%%%%%%%%%%%%%%%%%%%%%%%%%%%%%%%%%%%%%%%%%%%
\section{Introduction}
\label{sec:intro}
%%%%%%%%%%%%%%%%%%%%%%%%%%%%%%%%%%%%%%%%%%%%%%%%%%%%%%%%%%%%%%%%%%%%%%%%%%%%%%%

Stochastic multiplayer games (SMGs) 
are used to model reactive systems where
nondeterministic decisions are made by multiple players%
	~\cite{brazdil2018strategy,fremont2018reactive,neider2016automaton}.
SMGs extend probabilistic automata
by assigning a player
to each choice to be made in the game. This extension enables modeling of
complex systems where the behavior of players is unknown at
design time.
The \emph{strategy synthesis} problem aims to find a \emph{winning strategy},
\ie a~strategy that guarantees that a set of objectives (or winning conditions) is~satisfied~\cite{chen2013automatic,li2014synthesis}.
Algorithms for synthesis include, for instance, 
value iteration and strategy iteration techniques, 
where multiple reward-based objectives are satisfied~\cite{basset2015strategy,chen2013synthesis,kelmendi2018value}.
To tackle the state-space explosion problem,
\cite{wiltsche2015assume} presents an \emph{assume-guarantee} 
synthesis framework that relies on synthesizing strategies on the component
level first, %and later 
before composing them into a global winning strategy.
Mean-payoffs and ratio rewards are further investigated
in~\cite{basset2017compositional} to synthesize $\varepsilon$-optimal
strategies.
Formal tools that support strategy synthesis via SMGs include 
 \PRISM{}~\cite{chen2013prism,kwiatkowska2018prism} and
 {Uppaal Stratego}~\cite{david2015uppaal}.

SMGs are classified based on the number of players that can make
choices at each state.
In \emph{concurrent} games, more than one player is allowed to concurrently make choices at a given state.
Conversely, \emph{turn-based} games assign one player at most to each state.
Another classification considers the information available to
different players across the game~\cite{rasmusen1994games}.
\emph{Complete-information} games
(also known as
\emph{perfect-information} games~\cite{chatterjee2005semiperfect})
grant all players complete access to the information within
the game.
In \emph{symmetric} games, some information is equally
hidden from all players. 
On the contrary, \emph{asymmetric} games allow some players to have access
to more information 
than the others~\cite{rasmusen1994games}.

This work is motivated by security-aware systems in which stealthy adversarial actions are potentially hidden from the system, 
where the latter can probabilistically and intermittently gain full knowledge about 
the current state. 
While hidden-information games (HIGs) can be used to model such systems by using private variables to capture hidden information~\cite{chatterjee2005semiperfect},
standard model checkers can only synthesize strategies for (full-information) SMGs;
thus, demanding for alternative representations.
The equivalence between turn-based semi-perfect information games and concurrent perfect-information games was shown~\cite{chatterjee2005semiperfect}.
Since a player's strategy mainly rely on
full knowledge of~the game state~\cite{chen2013synthesis}, using SMGs for synthesis
produces strategies that 
may violate synthesis specifications in cases where required information is hidden from the player.
\emph{Partially-observable} stochastic games (POSGs) allow agents to have different belief states by incorporating uncertainty about both the current state
 and adversarial plans \cite{hansen2004dynamic}.
Techniques such as active sensing for online replanning \cite{fu2015integrating} and grid-based abstractions of belief spaces \cite{norman2015verification}
 were proposed to mitigate synthesis complexity arising from partial observability.
The notion of \emph{delaying actions} has been studied as means for gaining information about a game
 to improve future strategies~\cite{klein2015much,zimmermann2016delay}, but was not deployed as means for hiding information.

To this end,
we introduce delayed-action games (DAGs) --- a new class of games that simulate HIGs,
where information is hidden from one player by delaying the actions of the others.
The omission of private variables enables the use of off-the-shelf tools 
to implement and analyze DAG-based models.
We show how DAGs (under some mild and practical assumptions) can be decomposed into
subgames that can be independently explored,
reducing the time required for synthesis
 by employing parallel computation.
Moreover, we propose a DAG-based framework 
for strategy synthesis and analysis of security-aware systems.
Finally, we demonstrate the framework's applicability through a case study of security-aware planning for an unmanned-aerial vehicle (UAV) 
system prone to stealthy cyber attacks, where we develop a DAG-based system model and further synthesize strategies
 with strong probabilistic security~guarantees.

The paper is organized as follows.
\secref{sec:games} presents SMGs, HIGs, and 
problem formulation.
In \secref{sec:delayed}, we introduce DAGs 
and show that they can simulate HIGs.
\secref{sec:properties} 
proposes a DAG-based synthesis framework, which we use~for security-aware planning for UAVs in \secref{sec:casestudy},
before concluding the paper in \secref{sec:conclusion}.

%!TEX root = 00_CAV_2019_paper_78.tex
%%%%%%%%%%%%%%%%%%%%%%%%%%%%%%%%%%%%%%%%%%%%%%%%%%%%%%%%%%%%%%%%%%%%%%%%%%%%%%%
\section{Stochastic Games}
\label{sec:games}
%%%%%%%%%%%%%%%%%%%%%%%%%%%%%%%%%%%%%%%%%%%%%%%%%%%%%%%%%%%%%%%%%%%%%%%%%%%%%%%

In this section, we present turn-based stochastic games, which assume that all players have full information about the game state.
We then introduce hidden-information games and their private-variable~semantics. 

%==============================================================================
\subsubsection{Notation.}

We use $\ZP$ to denote the set of non-negative integers.
$\mathcal{P}(A)$ denotes the powerset of A (i.e., $2^A$).
A variable $v$ has a set of valuations $\Eval{v}$, where $\Val{v} \in \Eval{v}$ denotes one.
We use $\Sigma^*$ to denote the set of all finite
 words over alphabet $\Sigma$, including the empty word~$\epsilon$.
The mapping
$\Eff{} \! : \! \Sigma^* \! \times \! \Eval{v} \! \rightarrow \! \Eval{v}$
indicates the effect of a finite word on $\Val{v}$.
Finally, for general indexing, we use $s_i$ or $s^{(i)}$, for $i\in \ZP$, while 
$\mbox{PL}_\gamma$ denotes~\emph{Player~$\gamma$}.

%==============================================================================
\subsubsection{Turn-Based Stochastic Games (SMGs).}

SMGs can be used to model reactive systems that undergo both stochastic and nondeterministic transitions from one state to another.
In a \emph{turn-based} game,%
\footnote{%
	The term \emph{turn-based} indicates that
	at any state only one player can play an action.
	It does not necessarily imply that players take fair turns.}
actions can be taken at any state by at most one player.
Formally, an SMG can be defined as follows~\cite{baier2007stochastic,svorevnova2016quantitative,wiltsche2015assume}.
\begin{definition} [Turn-Based Stochastic Game]
A \emph{turn-based game (SMG)}
 with players $\Gamma = \{ \Pone, \Ptwo, \Psto \}$ is a tuple 
 $\SMG = \langle S, \allowbreak (S_\Pone, S_\Ptwo, S_\Psto), \allowbreak
 A,s_0,\delta \rangle$, where
\begin{itemize}
\item
	$S$ is a finite set of states, partitioned into $S_\Pone$, $S_\Ptwo$ and $S_\Psto$;
\item
	$A \Eq A_\Pone \cup A_\Ptwo \cup \{ \NoAction \} $ is a finite set of actions
	where $\NoAction$ is an empty action;
\item
	$s_0 \in S_\Ptwo$ is the initial state; and
\item
	$\delta: S \times A \times S \rightarrow [0,1]$ is a transition function,
	such that
	$\delta(s,a,s') \in \{1,0\}$, $\forall s \in S_\Pone \cup S_\Ptwo, a\in A $ and
	$s' \in S$,
	and $\delta(s,\NoAction,s') \in \left[0,1\right],~\forall s \in S_\Psto$ and
	$s' \in S_\Pone \cup S_\Ptwo$,
	where 
	$\sum_{s' \in S_\Pone \cup S_\Ptwo}{\delta(s,\NoAction,s')} = 1$ holds.
\end{itemize}
\end{definition}
For all $s \!\in S_\Pone \cup S_\Ptwo$ and $a \in A_\Pone \cup A_\Ptwo $, we write 
 {\larger[-1]$s \xrightarrow{a} s'$}
 if $\delta(s,a,s') \Eq 1$.
Similarly, for all $s \!\in\! S_{\Psto}$  we write 
 {\larger[-1]$s \xrightarrow{p} s'$}
if $s^\prime$ is randomly sampled with probability $p \Eq \delta (s,\NoAction,s')$.

%==============================================================================
\subsubsection{Hidden-Information Games.}

SMGs assume that all players 
 have full knowledge of the current state, and hence provide \emph{perfect-information} models~\cite{chatterjee2005semiperfect}.
In many applications, however, this assumption may not hold. A great example are security-aware
models where stealthy adversarial actions can be hidden from the system; \eg the system may not even be aware that it is under attack.
On the other hand, \emph{hidden-information} games (HIGs) refer to games
 where one player does not have complete access to (or knowledge of) the current state.
The notion of hidden information can be formalized with the use of \emph{private variables} (PVs)~\cite{chatterjee2005semiperfect}.
Specifically, a game state can be encoded using variables
 $v_\T$ and $v_\B$, representing 
 the true information, which is only known to \PL{1}, 
 and \PL{2} belief,~respectively.

\begin{definition} [Hidden-Information Game] %
\label{def:hig}
A \emph{hidden-information stochastic game (HIG)} with players $\Gamma = \{ \Pone, \Ptwo, \Psto \}$
 over a set of variables $V = \left\lbrace v_\T , v_\B \right\rbrace $
 is a tuple 
 $\HIG = \langle S, \allowbreak (S_\Pone, \allowbreak S_\Ptwo, S_\Psto), \allowbreak
 A,s_0, \beta, \delta \rangle$, where
\begin{itemize}
\item 
	set of states $S \subseteq \Eval{v_\T} \hspace{-2pt} \times \hspace{-2pt}\Eval{v_\B} \hspace{-2pt}\times\hspace{-2pt} \mathcal{P} \left(\Eval{v_\T}\right)\hspace{-2pt}  \times\hspace{-2pt} \Gamma$, partitioned in $ S_\Pone,  S_\Ptwo$, $ S_\Psto$;
\item 
	$A \Eq A_\Pone \cup A_\Ptwo \cup \{ \NoAction, \Reveal \} $ is a finite set of actions,
	where $\NoAction$ denotes an empty action, and $\Reveal$ is the action capturing \PL{2} attempt to reveal the true value~$v_\T$;
\item
	$s_0 \in S_\Ptwo$ is the initial state;
\item $\beta \colon A_\Ptwo \!\rightarrow\! \mathcal{P}\! \left(A_\Pone \right)$ is a function that defines
 the set of available \PL{1} actions, based on \PL{2}~action;~and
\item $ \delta \colon  S \times A \times  S \rightarrow [0,1] $ is a transition function such that 
		$ \delta ( s_\Pone,a, s_\Psto)=$
		$ \delta ( s_\Psto,a, s_\Pone) = 0$,
		and $\delta (s_\Ptwo,\Reveal, s_\Psto)$,
		$ \delta ( s_\Ptwo,a, s_\Pone)$,
		$ \delta ( s_\Pone,a, s_\Ptwo) \in \left\lbrace 0, 1 \right\rbrace$ for all
		$ s_\Pone \in  S_\Pone$, $ s_\Ptwo \in  S_\Ptwo$,
		$ s_\Psto \in  S_\Psto$ and $a \in A$,
		where $\sum_{s^\prime \in S_\Ptwo}{\delta(s_\Psto,\NoAction,s')} \Eq 1$.
\end{itemize}
\end{definition} %

In the above definition, $\delta$ only allows transitions  $s_\Pone$ to $s_\Ptwo$,
	$s_\Ptwo$ to $s_\Pone$ or $s_\Psto$, with $s_\Ptwo$  to  $s_\Psto$ conditioned by action $\Reveal$, 	
	and probabilistic transitions $s_\Psto$ to $s_\Ptwo$.
A game state can be written as $s = \left(t, u, \bspace, \pidx \right)$, but to simplify notation we use
 $s_{\pidx} \left(t, u, \bspace \right)$ instead, where
 $t \in \Eval{v_\T}$ is the \emph{true} value of the game,
 $u \in \Eval{v_\B}$ is \PL{2} current \emph{belief},
 $\bspace \in \mathcal{P}({\Eval{v_\T}}) \setminus \{ \emptyset \} $ is \PL{2} \emph{belief space},
 and $\pidx \in \Gamma$ is the current player's index.
When the truth is hidden from \PL{2}, the belief space $\Omega$ is the \emph{information set}~\cite{rasmusen1994games}, 
capturing \PL{2} knowledge about the possible~true~values.

\begin{wrapfigure}[8]{r}{0.452\textwidth}
	\centering
	\vspace{-2.2em}
	\includegraphics[width=0.44\columnwidth]{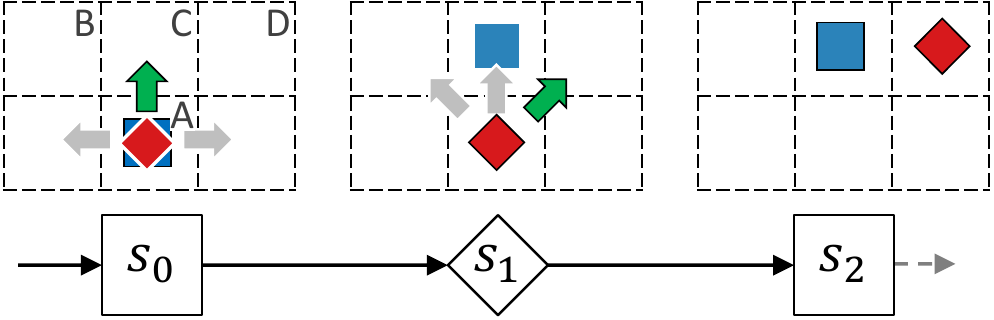}	
	\caption{The UAV belief (solid square) vs. the true value (solid diamond) of its location.}	
	\label{fig:player_roles}
\end{wrapfigure}

\vspace{\topsep} % DO NOT REMOVE
\noindent\textit{Example 1 (Belief vs. True Value).} \refstepcounter{example}\label{ex:running_example}
Our motivating example is a system that consists of a UAV and a human operator.
For localization, the UAV mainly relies on a GPS sensor that can be compromised to effectively steer the UAV away from its original path. 
While aggressive attacks can be detected,
 some may remain stealthy by introducing only bounded errors
 at each step~\cite{lesi_tecs17,pajic_csm17,mo2016performance,jovanov_tac19}.
For example, \figref{fig:player_roles} shows a UAV 
(\PL{2}) occupying zone \textsf{A} and flying north (N). 
An adversary (\PL{1}) can launch a stealthy attack targeting its GPS, 
 introducing a bounded error (NE, NW)
 to remain stealthy.
The set of stealthy actions available to the attacker depends on the preceding UAV
 action, which is captured by the function $\beta$,
 where~$\beta(\mathsf{N}) \Eq \lbrace \mathsf{NE}, \mathsf{N}, \mathsf{NW}\rbrace$.
Being unaware of the attack, the UAV believes that it is entering zone \textsf{C},~while the true new location is
 \textsf{D} due to the attack ($\mathsf{NE}$).
Initially, $\Val{v_{\T}} \Eq \Val{v_{\B}} \Eq z_A$,~and $\Omega \Eq \{ z_A \}$
 as the UAV is certain  it is in zone $z_A$.
In $s_2$, $\Val{v_{\B}} \Eq z_C$, yet $\Val{v_{\T}} \Eq z_D$.
Although $v_{\T}$ is hidden, 
 \PL{2} is aware that $\Val{v_{\T}}$ is in $\Omega \Eq \{ z_B, z_C, z_D \}$.
\vspace{\topsep} % DO NOT REMOVE
 
\noindent\textbf{HIG Semantics.} $\HIG$ semantics is described using the rules shown in \figref{fig:hig_rules},
where $\mathsf{H2}$ and $\mathsf{H3}$ capture \PL{2} and \PL{1} moves, respectively. 
The rule $\mathsf{H4}$ specifies that a \PL2 attempt $\Reveal$ to reveal the true value can succeed with probability $p_i$ where \PL{2} belief is updated
(\ie $u^\prime \!\Eq\! t$), and remains unchanged otherwise.

\begin{figure}[!t]
\centering
\scalebox{1}{\vbox{\input{tikz/FIG_HigSemantics}}}
\caption{%
	Semantic rules for an HIG.}
	\label{fig:hig_rules}
\end{figure}

\begin{example} [HIG Semantics]
\label{ex:hig_scenario}
Continuing~\exref{ex:running_example}, let us assume that the set of actions 
 $A_\Pone = A_\Ptwo =
	\lbrace \mathsf{N},\mathsf{S},\mathsf{E},\mathsf{W},\allowbreak
	\mathsf{NE},\mathsf{NW},\mathsf{SE},\mathsf{SW} \rbrace $, and that
 $\Reveal \Eq \mathsf{GT}$ is a geolocation task that attempts to reveal the true value of the game.%
 \footnote{A geolocation task is an attempt to localize the UAV by examining its camera feed.}
Now, consider the scenario illustrated in~\figref{fig:pv_game}.
At the initial state $s_0$, the UAV attempts to move north ($\mathsf{N}$),
progressing the game to the state $s_1$, where the adversary takes her turn by selecting
 an action from the set 
 $\beta(\mathsf{N}) = \lbrace \mathsf{NE}, \mathsf{N}, \mathsf{NW}\rbrace$. 
The players take turns until the UAV performs a geolocation task $\mathsf{GT}$, moving from the state $s_4$ to $s_5$.
With probability $p = \delta(s_5,\NoAction,s_6)$, the UAV detects its true~location and updates its belief accordingly (i.e., to $s_6$).
Otherwise, the belief remains the same~(i.e., equal to $s_4$).
\end{example}

\begin{figure}[tb]
	\centering  
	\includegraphics[width=1.0\textwidth]{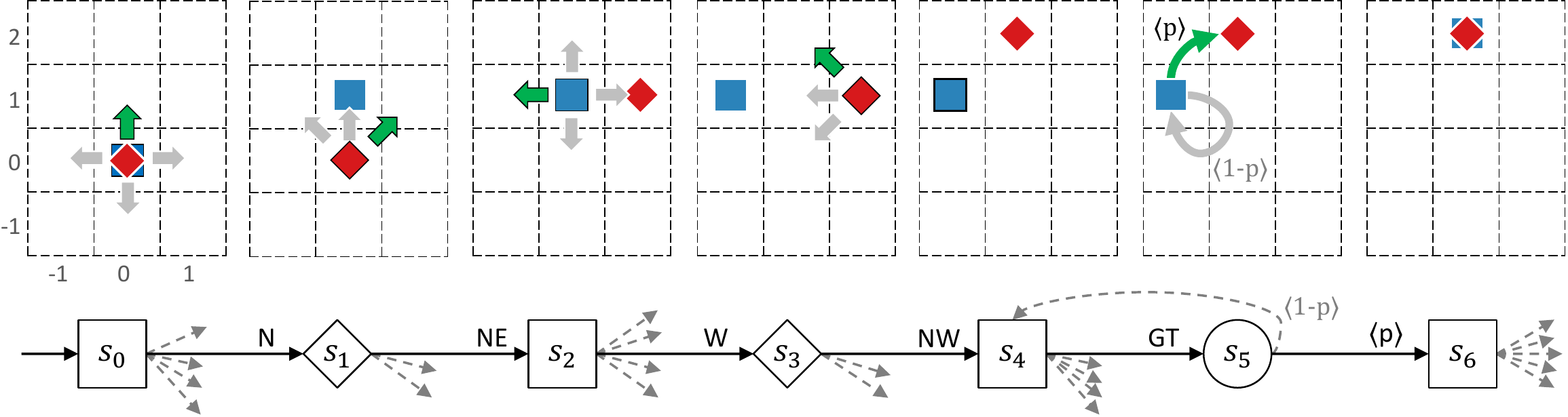}  
	\caption{%
		An example of the UAV motion in a 2D-grid map, modeled as an HIG.
		Solid squares represent the UAV belief,
		while solid diamonds represent the ground truth.
		The UAV action $\mathsf{GT}$ denotes performing a geolocation task.
		}
	\label{fig:pv_game}
\end{figure}

%==============================================================================
\subsubsection{Problem Formulation.}

Following the system described in \exref{ex:hig_scenario},
we now consider the composed HIG $\HIG = \M_{\Adv} \| \M_{\Uav} \| \M_{\Oas}$ 
 shown in \figref{fig:hig_system_model}; the HIG-based model incorporates standard models of a UAV ($\M_{\Uav}$),
 an adversary ($\M_{\Adv}$), and
 a geolocation-task advisory system ($\M_{\Oas}$) (\eg as introduced in~\cite{elfar2019security,feng2016synthesis}).
Here, the probability of a successful detection $p(v_\T,v_\B)$
 is a function of both the location
 the UAV believes to be its current location ($v_\B$) as well as the
 ground truth location that the UAV actually occupies ($v_\T$).
Reasoning about the flight plan using such model becomes problematic since the ground truth $v_\T$
 is inherently unknown to the UAV (i.e.,~\PL{2}), and thus so is $p(v_\T,v_\B)$.
Furthermore, such representation, where some information is hidden, is not supported by off-the-shelf
 SMG model checkers.
Consequently, for such HIGs, \emph{our goal is to find an alternative representation that is suitable for
strategy synthesis using off-the-shelf SMG~model-checkers.}

\begin{figure}[!t]
	\centering
	\includegraphics[scale=0.76]{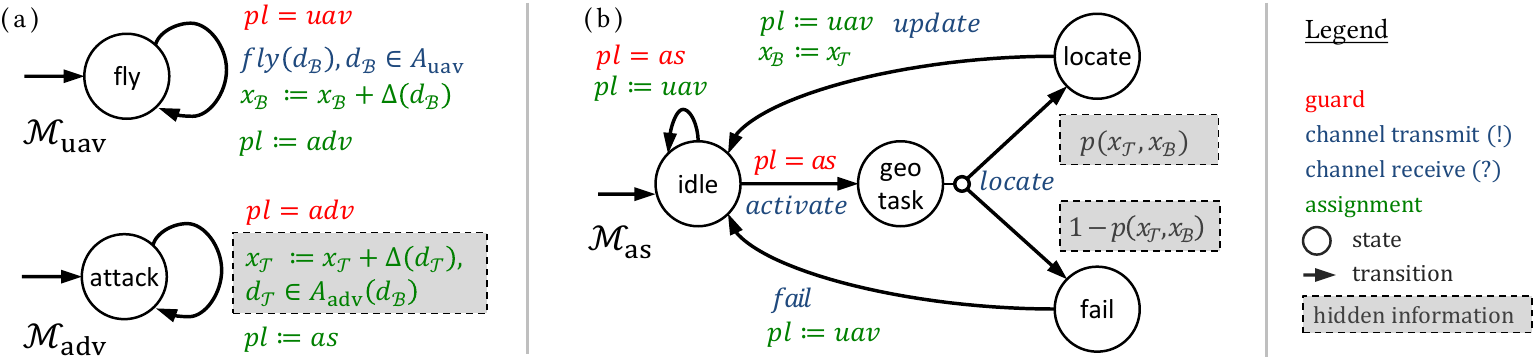}
	\caption{An example of an HIG-based system model comprised of the UAV ($\M_{\Uav}$),
 			 the adversary ($\M_{\Adv}$), and the AS ($\M_{\Oas}$).
 			 Framed information is hidden from the UAV-AS. 
 			 }
	\label{fig:hig_system_model}
\end{figure}

%!TEX root = 00_CAV_2019_paper_78.tex
%%%%%%%%%%%%%%%%%%%%%%%%%%%%%%%%%%%%%%%%%%%%%%%%%%%%%%%%%%%%%%%%%%%%%%%%%%%%%%%
\section{Delayed-Action Games}
\label{sec:delayed}
%%%%%%%%%%%%%%%%%%%%%%%%%%%%%%%%%%%%%%%%%%%%%%%%%%%%%%%%%%%%%%%%%%%%%%%%%%%%%%%

In this section, we provide an alternative representation of HIGs that eliminates the use of private variables --- we introduce Delayed-Action Games (DAGs) 
 that exploit the notion of \emph{delayed actions}.
Furthermore, we  show that for any HIG, a DAG that simulates the former can be constructed.

%==============================================================================
\subsubsection{Delayed Actions.}

Informally, a DAG reconstructs an HIG such that actions of \PL{1} (the player with access to perfect information) follow 
the actions of \PL{2}, \ie \PL{1} actions are \emph{delayed}.
This rearrangement of the players' actions provides a means to hide information from \PL{2} without the use of private variables, since in this case, at \PL{2} states, \PL{1} actions have not occurred yet.
In this way, \PL{2} can act as though she has complete information at the moment she makes
 her decision, as the future state has not yet happened and so cannot be known.
In essence, the formalism can be seen as a partial ordering of the players' actions,
 exploiting the (partial) superposition property that a wide class of physical systems exhibit.
To demonstrate this notion, let us consider DAG modeling on our running example.

\begin{example} [Delaying Actions]
\label{ex:da_game}
\figref{fig:da_game} depicts the (HIG-based) scenario  from \figref{fig:pv_game}, but in the corresponding DAG,
where the UAV actions are performed~first (in $\hat{s}_0,\hat{s}_1,\hat{s}_2$), 
followed by the adversary delayed actions (in $\hat{s}_3,\hat{s}_4$).
Note that, in the DAG model, at the time the UAV executed its actions ($\hat{s}_0,\hat{s}_1,\hat{s}_2$) the adversary actions had not occurred (yet).
Moreover, $\hat{s}_0$ and $\hat{s}_6$ (\figref{fig:da_game}) share the same
belief and true values as $s_0$ and $s_6$ (\figref{fig:pv_game}), respectively, though
 the transient states do not exactly match. This will be used to show the relationship between the~games.
\end{example}

\begin{figure}[!t]
	\centering 
	\includegraphics[width=1.0\textwidth]{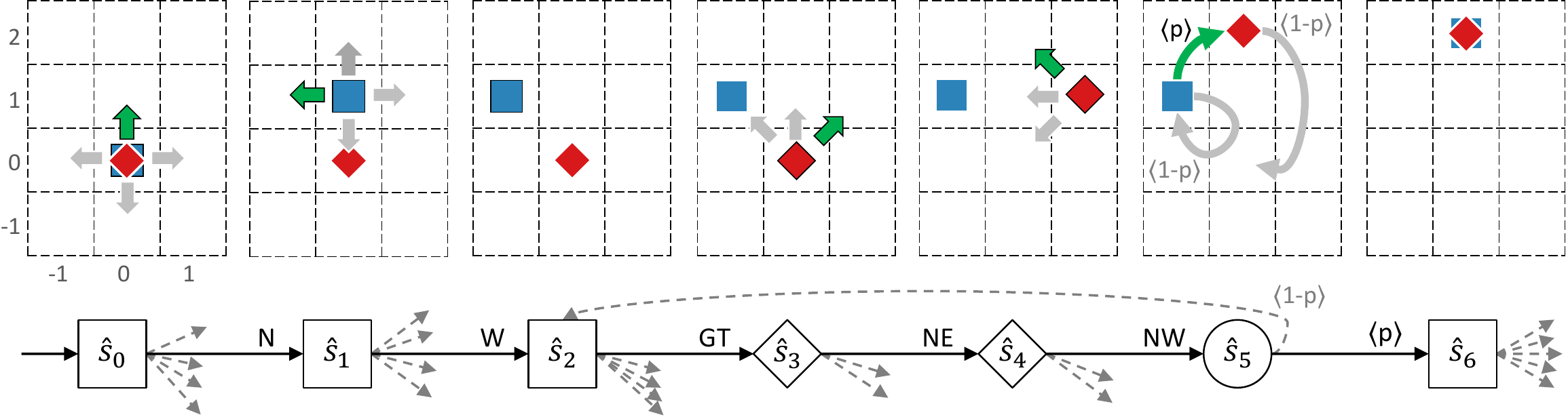}  
	\caption{%
	The same scenario as in~\figref{fig:pv_game}, modeled as a DAG.
	Solid squares represent UAV belief,
		while solid diamonds represent the ground truth.
		The UAV action $\mathsf{GT}$ denotes performing a geolocation task.}
	\label{fig:da_game}
\end{figure}

The advantage of this approach is twofold.
First, the elimination of private variables enables simulation of an HIG using
 a full-information game.
Thus, the formulation of the strategy synthesis problem using off-the-shelf SMG-based tools becomes feasible.
 %since \PL{2} actions precede the knowledge of any specific \PL{1} actions taken.
In particular, a \PL{2} synthesized strategy becomes dependent on the knowledge of \PL{1}
 behavior (possible actions), rather than the specific (hidden) actions.
We  formalize a DAG as follows.
\begin{definition}[Delayed-Action Game]\label{def:dag}
A \emph{DAG} of an HIG 
$\HIG = \langle S, (S_\Pone, \allowbreak S_\Ptwo, S_\Psto),  A,s_0, \beta, \delta \rangle$,
with players $\Gamma = \{ \Pone, \Ptwo, \Psto \}$
over a set of variables $V \Eq \left\lbrace v_\T , v_\B \right\rbrace $ is a tuple 
$\DAG = \langle \hat{S},\allowbreak(\hat{S}_\Pone, \hat{S}_\Ptwo,
	\hat{S}_\Psto), \allowbreak A, \hat{s}_0, \beta, \hat{\delta} \rangle $
where
\begin{itemize}
\item $\hat{S} \subseteq \Eval{v_\T} \times \Eval{v_\B} \times A_{\Ptwo}^* \times \ZP \times \Gamma$ is the set of states, partitioned into $\hat{S}_\Pone, \hat{S}_\Ptwo$ and $\hat{S}_\Psto$;
\item $\hat{s}_0 \in \hat{S}_\Ptwo$ is the initial state; and
\item $\hat{\delta} \colon \hat{S} \times A \times \hat{S} \rightarrow [0,1] $ is a transition function such that 
		$\hat{\delta} (\hat{s}_\Ptwo,a,\hat{s}_\Psto) \Eq $
		$\hat{\delta} (\hat{s}_\Pone,a,\hat{s}_\Ptwo) \Eq $
		$\hat{\delta} (\hat{s}_\Psto,a,\hat{s}_\Pone) \Eq 0$, and
		$\hat{\delta} (\hat{s}_\Ptwo,a,\hat{s}_\Ptwo) \in \left\lbrace 0, 1 \right\rbrace$,
		$\hat{\delta} (\hat{s}_\Ptwo,\Reveal,\hat{s}_\Pone) \in \left\lbrace 0, 1 \right\rbrace$,
		$\hat{\delta} (\hat{s}_\Pone,a,\hat{s}_\Pone) \in \left\lbrace 0, 1 \right\rbrace$,
		$\hat{\delta} (\hat{s}_\Pone,a,\hat{s}_\Psto) \in \left\lbrace 0, 1 \right\rbrace$,
		for all $\hat{s}_\Pone \in \hat{S}_\Pone$,
		$\hat{s}_\Ptwo \in \hat{S}_\Ptwo$, $\hat{s}_\Psto \in \hat{S}_\Psto$ and $a \in A$,
		where $\sum_{\hat{s}^\prime \in \hat{S}_\Ptwo}{\delta(\hat{s}_\Psto,a,\hat{s}')} \Eq 1$.
\end{itemize}
\end{definition}
Note that, in contrast to transition function $\delta$ in HIG $\HIG$, $\hat{\delta}$ in DAG $\DAG$ only allows transitions  
	$\hat{s}_\Ptwo$ to $\hat{s}_\Ptwo$ or $\hat{s}_\Pone$, as well as 
	$\hat{s}_\Pone$ to $\hat{s}_\Pone$ or $\hat{s}_\Psto$,
	and probabilistic transitions $\hat{s}_\Psto$ to $\hat{s}_\Ptwo$; 
	also note that 
 $\hat{s}_\Ptwo$  to  $\hat{s}_\Pone$ is conditioned by the action $\Reveal$.

%==============================================================================
\subsubsection{DAG Semantics.}
A DAG state is a tuple $\hat{s} \Eq \left(\hat{t}, \hat{u}, w, j, \pidx \right)$, which for simplicity we
 shorthand as $\hat{s}_\pidx \left(\hat{t}, \hat{u}, w, j \right)$, 
 where $\hat{t} \in \Eval{v_\T}$ is the last known true value,
 $\hat{u} \in \Eval{v_\B}$ is \PL{2} belief,
 $w \in A_\Ptwo^*$ captures \PL{2} actions taken since the last known true value,
 $ j \in \ZP $ is an index on $w$,
 and $\pidx \in \Gamma $ is the current player index.
The game transitions are defined using the semantic rules from~\figref{fig:dag_rules}.
Note that \PL{2} can execute multiple moves (\ie actions) before executing $\Reveal$ to attempt to reveal the true value ($\mathsf{D2}$), moving to a \PL{1} state where \PL{1} executes  all her delayed actions before reaching a `revealing' state $\hat{s}_\Psto$
($\mathsf{D3}$).
Finally, the revealing attempt can succeed with probability $p_i$ where \PL{2} belief is updated
 (\ie $\hat{u}^\prime \Eq \hat{t}\,$), or otherwise remains~unchanged~($\mathsf{D4}$).

\begin{figure}[!t]
\centering
\scalebox{0.98}{\vbox{\input{tikz/FIG_DagSemantics}}}
\caption{%
	Semantic rules for DAGs. }
	\label{fig:dag_rules}
\end{figure}

In both $\HIG$ and $\DAG$, we label states where all players have full knowledge of the current state
as \emph{proper}. We also say that two states are similar if they agree on the belief, and equivalent if they agree on both the belief and ground truth.

\begin{definition}[States]
Let $s_\gamma (t,u,\Omega)\in S$ and $\hat{s}_{\hat{\gamma}} (\hat{t}, \hat{u}, w,j) \in \hat{S}$. We say:
\begin{itemize}
\item $s_\gamma$ is \emph{proper} iff $\bspace = \{t\}$, denoted by $s_\gamma \in \Proper(\HIG)$.
\item $\hat{s}_{\hat{\gamma}}$ is \emph{proper} iff $w = \epsilon$, denoted by $\hat{s}_{\hat{\gamma}} \in \Proper(\DAG)$.
\item $s_\gamma$ and $\hat{s}_{\hat{\gamma}}$ are \emph{similar} iff $\hat{u} = u$, $\hat{t} \in \Omega$, and $\gamma=\hat{\gamma}$, denoted by $s _\gamma\sim \hat{s}_{\hat{\gamma}}$.
\item $s_\gamma$, $\hat{s}_{\hat{\gamma}}$ are \emph{equivalent} iff $t = \hat{t}$, $u = \hat{u}$, $w=\epsilon$,
 and $\gamma=\hat{\gamma}$, denoted by $s_\gamma \simeq \hat{s}_{\hat{\gamma}}$.
\end{itemize}
\end{definition}
From the above definition, we have that $s \simeq \hat{s} \implies s \in \Proper(\HIG), \hat{s} \in \Proper(\DAG)$. 
We now define \emph{execution fragments}, possible progressions from a state~to~another.

\begin{definition}[Execution Fragment]
An \emph{execution fragment} (of either an SMG, DAG or HIG) is a finite sequence of states, actions and probabilities\\
$\varrho = s_0 a_1 p_1 s_1 a_2 p_2 s_2 \ldots a_n p_n s_n $
such that
 {\larger[-1.5]$
	( s_i \xrightarrow{a_{i+1}}s_{i+1}        )
	\lor
	( s_i \xrightarrow{\prob{p_{i+1}}}s_{i+1} ),
	\forall i \geq 0
 $}.\footnote{For deterministic transitions, $p=1$, hence omitted from $\varrho$ for readability.}
\end{definition}
We use $\First(\varrho)$ and $\Last(\varrho)$
 to refer to the first and last states of $\varrho$, respectively.
If both states are proper, we say that $\varrho$ is \emph{proper} as well,
 denoted by $\varrho \in \Proper(\HIG)$.%
 \footnote{%
 An execution fragment lives in the transition system (TS), \ie
 $\varrho \in \Proper(\mathrm{TS}(\GX))$. We omit $\mathrm{TS}$ for readability.}
Moreover, $\varrho$ is \emph{deterministic} if no probabilities appear in the sequence.

\begin{definition}[Move]
A \emph{move} $m_\gamma$ of an execution $\varrho$ from state $s\in\varrho$,
  denoted by $\Move_\gamma(s,\varrho)$,
  is a sequence of actions $a_1 a_2 \ldots a_i \in A_\gamma^*$
  that player $\gamma$ performs in $\varrho$ starting from $s$.
\end{definition}
By omitting the player index
 we refer to the moves of all players.
To simplify notation, we use $\Move(\varrho)$ as a short notation for 
 $\Move(\First(\varrho),\varrho)$.
We write $(m)(\First(\varrho)) = \Last(\varrho)$ to denote that
the execution of move $m$ from the $\First(\varrho)$ leads to the $\Last(\varrho)$.
This allows us to now define the \emph{delay operator} as follows.
\begin{definition}[Delay Operator]\label{def:delay_operator}
For an $\HIG$, let $m = \Move(\varrho) = a_1 b_1 \ldots a_n b_n \Reveal$ be a move
 for some deterministic $\varrho \in \Ts(\HIG) $, where
 $a_1 ... a_n \in A_\Ptwo^*, b_1 ... b_n \in A_\Pone^*$.
The delay operator, denoted by
 $\overline{m}$, is defined by the rule
 $\overline{m} = a_1 \ldots a_n \Reveal b_1 \ldots b_n$.
\end{definition}
Intuitively, the delay operator shifts \PL{1} actions to the right of \PL{2} actions up until the next probabilistic state. For example,
\begin{align*}
\begin{matrix}
\mbox{if} &\rho =
	&s_{\Ptwo}^{(0)}
	&\!\xrightarrow{a_1}     \! &s_{\Pone}^{(1)}
	&\!\xrightarrow{b_2}     \! &s_{\Ptwo}^{(2)} 
	&\!\xrightarrow{\Reveal} \! &s_{\Psto}^{(3)}
	&\!\xrightarrow{p_3}     \! &s_{\Ptwo}^{(4)}
	&\!\xrightarrow{a_4}     \! &s_{\Pone}^{(5)}
	&\!\xrightarrow{b_5}     \! &s_{\Ptwo}^{(6)}
	&\!\xrightarrow{a_6}     \! &s_{\Pone}^{(7)}
	&\!\xrightarrow{b_7}     \! &s_{\Ptwo}^{(8)}
	\\
\mbox{then } & m =
	& 
	&a_1 &
	&\RememberNode{12}{b_2} &
	&\RememberNode{13}{\Reveal} &
	&\NoAction &
	&a_4 &
	&\RememberNode{15}{b_5} &
	&\RememberNode{16}{a_6} & 
	&b_7, &	\\
\mbox{and} &\overline{m} =
	& 
	&a_1 &
	&\RememberNode{23}{\Reveal} &
	&\RememberNode{22}{b_2} &
	&\NoAction &
	&a_4 &
	&\RememberNode{26}{a_6} &
	&\RememberNode{25}{b_5} & 
	&b_7. &
\end{matrix}
\end{align*}

% Overlay arrows
\begin{tikzpicture}[overlay,remember picture]
	\draw [-latex,shorten >=2pt,shorten <=2pt] (12) -- (22);
	\draw [-latex,shorten >=2pt,shorten <=2pt] (13) -- (23);
	\draw [-latex,shorten >=2pt,shorten <=2pt] (15) -- (25);
	\draw [-latex,shorten >=2pt,shorten <=2pt] (16) -- (26);
\end{tikzpicture}

%==============================================================================
\subsubsection{Simulation Relation.}
Given an HIG $\HIG$, we first define the corresponding DAG $\DAG$.
\begin{definition}[Correspondence]
Given an HIG $\HIG$, a corresponding DAG $\DAG = \mathfrak{D}[\HIG]$ is a DAG that follows the semantic rules displayed in~\figref{fig:hig_to_dag}.
\end{definition}

\begin{figure}[!t]
\centering
\scalebox{0.98}{\vbox{\input{tikz/FIG_Mapping}}}
\caption{%
	Semantic rules for HIG-to-DAG transformation.}
	\label{fig:hig_to_dag}
\end{figure}

For the rest of this section, we consider $\DAG = \mathfrak{D} [\HIG]$, and use
$\varrho \in \Ts(\HIG)$ and $\hat{\varrho} \in \Ts(\DAG)$ to denote two execution fragments of the HIG and DAG, respectively.
We say that $\varrho$ and $\hat{\varrho}$ are \emph{similar},
  denoted by $\varrho \sim \hat{\varrho}$,
  iff
	$\First(\varrho) \simeq \First(\hat{\varrho})$,
	$\Last(\varrho) \sim \Last(\hat{\varrho})$, and 
	$\overline{\Move(\varrho)} = \Move(\hat{\varrho})$.

\begin{definition}[Game Proper Simulation]
\label{def:game_proper_simulation}
 A game $\DAG$ properly simulates $\HIG$,
  denoted by $\DAG \simulate \HIG$,
  iff
  $\forall\varrho \in \Proper(\HIG)$,
  $\exists \hat{\varrho} \in \Proper(\DAG)$
  such that $\varrho \sim \hat{\varrho}$.
\end{definition}
Before proving the existence of the simulation relation, we first show that if a move is executed on two equivalent states, then the terminal states are similar.

\begin{lemma}[Terminal States Similarity]
\label{lem:terminal_states} %
For any $s_0 \!\simeq\! \hat{s}_0$
and~a~deterministic $\varrho \!\in\! \Ts(\HIG)$
where $\First(\varrho) \Eq s_0$, $\Last(\varrho) \in S_{\Ptwo}$,
 then
 $\Last(\varrho) \!\sim\! \left(\overline{\Move(\varrho)}\right)\! \left(\hat{s}_0\right)$ holds.
\end{lemma}

\begin{proof}
Let $\Last(\varrho_i) = s^{(i)}_{\gamma_i} (t_i,u_i,\bspace_i)$ and 
 $\left(\overline{\Move(\varrho_i)}\right) \left(\hat{s}_0\right) = \hat{s}^{(i)}_{\hat{\gamma}_i} (\hat{t}_i,\hat{u}_i,w_i,j_i)$,
where {$\Move(\varrho_i) = a_1 b_1 ... a_i b_i \Reveal $}.
We then write  $\overline{\Move(\varrho)}=a_1 ... a_i \Reveal b_1 ... b_i$.
We use induction over $i$ as follows:
\begin{itemize}
\item Base $(i \Eq 0)$:
	$\varrho_0 \Eq s_0 $
	$\implies s^{(0)} \simeq \hat{s}^{(0)} $ where $u_0=\hat{u}_0$
	and $t_0 \Eq \hat{t}_0$.
\item Induction $(i \!>\! 0)$: Assume that the claim holds for {$\Move(\varrho_{i-1}) \Eq a_1 b_1 ... a_{i-1} b_{i-1} \Reveal $},
\ie $u_{i-1} \Eq \hat{u}_{i-1}$ and $\hat{t}_{i-1} \!\in\! \Omega_{i-1}$. 
For $\varrho_i$ we have that
	$u_i \Eq \Eff{a_i,u_{i-1}}$ and
	$\hat{u}_i \Eq \Eff{a_i,\hat{u}_{i-1}}$.
Also, $t_i \Eq \Eff{b_i,t_{i-1}} \in \Omega_{i}$ and
	$\hat{t}_i \Eq \Eff{b_i,\hat{t}_{i-1}} $.
Hence, $ u_i \Eq \hat{u}_i$, $\hat{t}_i \in \Omega_i$ and
$\hat{\gamma}_i \Eq \gamma_i \Eq \Psto$. Thus, $s^{(i)} \sim \hat{s}^{(i)}$ holds. The same can be shown for
$\Move(\varrho) = a_1 b_1 ... a_i b_i $ where no $\Reveal$ occurs. \qed
\end{itemize}
\end{proof}

\begin{theorem}[Probabilistic Simulation] %---------------------->>
\label{thm:prob_bisim}
For any $s_0 \simeq \hat{s}_0$ 
and $\varrho \in \Proper(\HIG)$ where $\First(\varrho) = s_0$,
it holds that
 $$
 	\mathrm{Pr} \left[ \Last(\varrho) = s^{\prime} \right] = 
  	\mathrm{Pr} \left[ \left(\overline{\Move(\varrho)}\right) \left(\hat{s}_{0}\right) = \hat{s}^{\prime} \right]
 \quad
  %%%\forall \left(s^\prime, \hat{s}^\prime \right) \in \simeq.
\forall s^\prime,\hat{s}^\prime ~~s.t.~~ s^\prime \simeq \hat{s}^\prime.
 $$
\end{theorem}
\begin{proof}
We can rewrite $\varrho$ as
$
	\varrho = \varrho_0 \xrightarrow{p_1} \varrho_1 \cdots \varrho_{n-1} \xrightarrow{p_n} s_{\Ptwo}^{(n)}
$, 
 where $\varrho_0, \varrho_1,\ldots, \varrho_{n-1}$ are deterministic.
 Let $\First(\varrho_i) = s_\Ptwo^{(i)} (t_i,u_i,\Omega_i)$,
 $\Last(\varrho_i) = s_\Psto^{(i)} (t_i^\prime,u_i^\prime,\Omega_i^\prime)$,
 and $\left(\overline{\Move(\varrho)}\right) \! \left(\hat{s}_0 \right) 
 				\Eq \hat{s}^{(n)} (\hat{t}_n, \hat{u}_n, w_n, j_n)$.
We use induction over $n$ as follows:
\begin{itemize}
\item Base ($n \Eq 0$): for $\varrho$ to be deterministic and proper, $\varrho \Eq \varrho_0 \Eq s^{(0)} $ holds.
\item Case ($n \Eq 1$): $p_1 = p(t_0^\prime,u_0^\prime) $.
	From \lemref{lem:terminal_states}, $\hat{u}_1 \Eq u_1 $ and $\hat{t}_1 \Eq t_1 $.
	Hence,
	{\larger[-1]$ \mathrm{Pr} \left[ \Last(\varrho) \Eq s_\Ptwo^{(1)} \right] = 
  	\mathrm{Pr} \left[ \left(\overline{\Move(\varrho)}\right) \left(\hat{s}_{0}\right) \Eq \hat{s}_\Ptwo^{(1)} \right] \Eq p(t_0^\prime,u_0^\prime) $}
  	and $s_\Ptwo^{(1)} \simeq \hat{s}_\Ptwo^{(1)}$.
\item Induction ($n \!>\! 1$): 
	It is straightforward to infer that $ p_n \Eq p\left(t_{n-1}^\prime , u_{n-1}^\prime \right) $,~hence 
 	{$ \mathrm{Pr} \! \left[ \Last(\varrho) \Eq s_\Ptwo^{(n)} \right] = 
  	\mathrm{Pr} \! \left[ \left(\overline{\Move(\varrho)}\right) \! \left(\hat{s}^{(0)}\right) \Eq \hat{s}^{(n)} \right] = P$},
  	and  $s_\Ptwo^{(n)} \simeq \hat{s}_\Ptwo^{(n)}$.
	 \qed
\end{itemize}
\end{proof}
Note that in case of multiple $\Reveal$ attempts, the above probability $P$ satisfies
%\vspace{-3pt}
$$
P = \prod_{i=1}^{n} \sum_{j=1}^{m_i} p_i\left(t_{i-1}^\prime , u_{i-1}^\prime\right)
	\left( 1- p_{i-1}\left(t_{i-1}^\prime , u_{i-1}^\prime \right) \right)^{(j-1)},
$$
where $m_i$ is the number of $\Reveal$ attempts at stage $i$.
Finally, since \thmref{thm:prob_bisim} imposes no constraints on $\Move(\varrho)$,
a DAG can simulate all proper executions that exist in the corresponding HIG.

\begin{theorem}[DAG-HIG Simulation]
\label{thm:pv_da_bisim}
For any HIG $\HIG$
there exists a DAG $\DAG = \mathfrak{D}[\HIG]$ such that $\DAG \simulate \HIG $ (as defined in Def.~\ref{def:game_proper_simulation}).
\end{theorem}

%!TEX root = 00_CAV_2019_paper_78.tex
%%%%%%%%%%%%%%%%%%%%%%%%%%%%%%%%%%%%%%%%%%%%%%%%%%%%%%%%%%%%%%%%%%%%%%%%%%%%%%%
\section{Properties of DAG and DAG-based Synthesis }
\label{sec:properties}
%%%%%%%%%%%%%%%%%%%%%%%%%%%%%%%%%%%%%%%%%%%%%%%%%%%%%%%%%%%%%%%%%%%%%%%%%%%%%%%

We here discuss DAG features,
 including how it can be decomposed into subgames by restricting the simulation to finite executions,
 and the preservation of safety properties,
 before proposing a DAG-based synthesis framework.

%==============================================================================
\subsubsection{Transitions.}

In DAGs, nondeterministic actions of different players underline different semantics.
Specifically, \PL{1} nondeterminism captures what is known about the adversarial behavior,
 rather than exact actions, where \PL{1} actions are constrained by the earlier \PL{2} action.
Conversely, \PL{2} nondeterminism~abstracts the player's decisions.
This distinction reflects how DAGs can be used~for strategy synthesis under hidden information.
To illustrate this, suppose that a strategy $\pi_\Ptwo$ is to be obtained based
 on a worst-case scenario.
In that case,
 the game is explored for all possible adversarial behaviors.
Yet, if a strategy~$\pi_\Pone$ is known about \PL{1},
  a counter strategy $\pi_\Ptwo$ can be found by constructing  $\DAG^{\pi_\Pone}$. 

Probabilistic behaviors in DAGs are captured by \PLs,
which is characterized by the transition function
 $\hat{\delta} \colon \hat{S}_\Psto \times \hat{S}_\Ptwo \rightarrow [0,1]$.
The specific definition of $\hat{\delta}$ depends on the modeled system. 
For instance, if the transition function (i.e., the probability) is state-independent,
\ie $ \hat{\delta}(\hat{s}_\Psto, \hat{s}_\Ptwo) = c , c \in [0,1]$,
 the obtained model becomes trivial.
Yet, with a state-dependent transition function,
\ie $\hat{\delta}(\hat{s}_\Psto,\hat{s}_\Ptwo) = p(\hat{t},\hat{u})$,
 the probability that \PL{2} successfully reveals the true value depends on both the belief and the true value,
and the transition function can then be realized since $\hat{s}_\Psto$ holds both $\hat{t}$ and $\hat{u}$.

%==============================================================================
\subsubsection{Decomposition.}

Consider an execution $\hat{\varrho}^* \Eq \hat{s}_0 a_1 \hat{s}_1 a_2 \hat{s}_2 \ldots$
 that describes a scenario where \PL{2} performs infinitely many actions with no attempt to reveal the true value. 
To simulate $\hat{\varrho}^*$, the word $w$ needs to infinitely grow.
Since we are interested in finite executions,
 we impose \emph{stopping criteria} on the DAG, such that the game is \emph{trapped} whenever $|w| \Eq h_{\max}$ is true,
 where $h_{\max} \in \mathbb{N}$ is an \emph{upper horizon}.
We formalize the stopping criteria as a deterministic finite automaton (DFA) that, when composed with the DAG, traps the game
 whenever the stopping criteria hold.
Note that imposing an upper horizon by itself is not a sufficient criterion for a DAG to be considered a stopping game~\cite{chen2013stochastic}.
Conversely, consider a proper (and hence finite) execution
 $\hat{\varrho} \Eq \hat{s}_0 a_1 \ldots \hat{s}^\prime$,
 where $\hat{s}_0, \hat{s}^\prime \in \Proper(\DAG)$.
From \defref{def:game_proper_simulation}, it follows that a DAG initial state is strictly proper,
\ie $\hat{s}_0 \in \Proper(\DAG)$.
Hence, when $\hat{s}^\prime$ is reached, the game can be seen as if it is \emph{repeated} with a new initial state
 $\hat{s}^\prime$.
Consequently, a DAG game (complemented with stopping criteria)  can be decomposed into a (possibly infinite) countable
set of \emph{subgames} that have the same structure yet different initial states.

\begin{definition}[DAG Subgames]\label{def:dag_subgames}
The \emph{subgames} of a $\DAG$ are defined by the set
{\larger[-1]$
 \left\lbrace
	\SUB_i \;\middle|\;
	\SUB_i =
	\left\langle \hat{S}^{(i)},\allowbreak (\hat{S}_\Pone^{(i)},
	\hat{S}_\Ptwo^{(i)}, \hat{S}_\Psto^{(i)}),\allowbreak
	A,\hat{s}_0^{(i)}, \hat{\delta}^{(i)} \right\rangle ,\,
	i \in \ZP
  \right\rbrace,
$}
where
$ \hat{S} = \bigcup_{i}^{}{\hat{S}^{(i)}} $;
$ \hat{S}_{\pidx} = \bigcup_{i}^{}{\hat{S}_{\pidx}^{(i)}}
		\;  
		\forall \pidx \in \Gamma $; and
$ \hat{s}_0^{(i)} = \hat{s}_{\Ptwo}^{(i)}
		\; \suchthat \;
		\hat{s}_{\Ptwo}^{(i)} \in \Proper(\DAG^{(i)})
		\; , \;
		\hat{s}_{\Ptwo}^{(i)} \neq \hat{s}_{\Ptwo}^{(j)}
		\;
		\forall i,j \in \ZP
		$.
\end{definition}

Intuitively, each subgame either reaches a proper state (representing the initial state of another subgame) or terminates by an upper horizon.
This decomposition allows for the independent (and parallel) analysis of individual subgames,
 drastically reducing both the time required for synthesis and the explored state space, and hence improving scalability.
An example of this decompositional approach is provided in~\secref{sec:casestudy}.

%==============================================================================
\subsubsection{Preservation of safety properties.}

In DAGs, the action $\Reveal$ denotes a transition from \PL{2} to \PL{1} states and thus the execution of any delayed actions. 
While this action can simply describe a revealing attempt, it can also serve as a \emph{what-if} analysis of how the true value 
 may evolve at stage $i$ of a subgame.
We refer to an execution of the second type as a \emph{hypothetical branch}, where $\Hyp(\hat{\varrho},h)$ denotes the set of hypothetical branches from $\hat{\varrho}$
 at stage $h \in \{ 1, \ldots , n \} $. 
Let $L_{\mathrm{safe}} (s)$
 be a labeling function denoting if a state is safe.
The formula
 $ \Phi_{\mathrm{safe}} \coloneqq \left[ \always \, \Safe \right]$ is satisfied by an execution
 $\varrho$ in HIG iff all $ s(t,u,\Omega) \in \varrho$ are safe. 

Now, consider $\hat{\varrho}$ of the DAG, with $\hat{\varrho} \sim \varrho$. We identify the following three~cases: 
\\(a)~$L_{\mathrm{safe}} (s)$ depends only on the belief $u$, then $\varrho \models \Phi_{\mathrm{safe}}$ iff all $\hat{s}_\Ptwo \in \hat{\varrho}$ are~safe;
\\(b)~$L_{\mathrm{safe}} (s)$ depends only on the true value $t$, then $\varrho \models \Phi_{\mathrm{safe}}$ iff all $\hat{s}_\Pone \in \Hyp(\hat{\varrho},n)$ are safe; and
\\(c)~$L_{\mathrm{safe}} (s)$ depends on both the true value $t$ and belief $u$, then $\varrho \models \Phi_{\mathrm{safe}}$~iff~$\Last(\hat{\varrho}_h)$ is safe for all $ \hat{\varrho}_h \!\in\! \Hyp(\hat{\varrho},h), h \!\in\! \{ 1, ... , n \},$
where $n$ is the number of~\PL{2}~actions.
\\Taking into account such relations,
 both safety (\eg never encounter a hazard) and distance-based requirements (\eg never exceed a subgame horizon) can be specified when using DAGs for synthesis,
 to ensure their satisfaction in the original model.
This can be generalized to other reward-based synthesis objectives, which will be part of our future efforts that we discuss in~\secref{sec:conclusion}.

%==============================================================================
\subsubsection{Synthesis Framework.}

We here propose a framework for strategy synthesis using DAGs,
 which is summarized in~\figref{fig:framework}.
We start by formulating the automata $\M_\Pone$, $\M_\Ptwo$ and $\M_\Psto$,
 representing \PL{1}, \PL{2} and \PLs  ~abstract behaviors, respectively.
Next, a FIFO memory stack  $(m_i)_{i=1}^{n} \in A_{\Ptwo}^n$ is implemented using
 two automata $\M_{\Mrd}$ and $\M_{\Mwr}$ to perform reading and writing operations, respectively.\footnote{Specific implementation details are described in~\secref{sec:casestudy}.}
The DAG $\DAG$ is constructed by following \algref{alg:dag_construction}.
The game starts with \PL2 moves until she executes a revealing attempt $\Reveal$, allowing \PL1 to play her delayed actions.
Once an end criterion is met, the game terminates,
resembling conditions such as `running out of fuel' or `reaching map boundaries'.

\begin{figure}[tb]
	\centering
	\includegraphics[width=0.95\columnwidth]{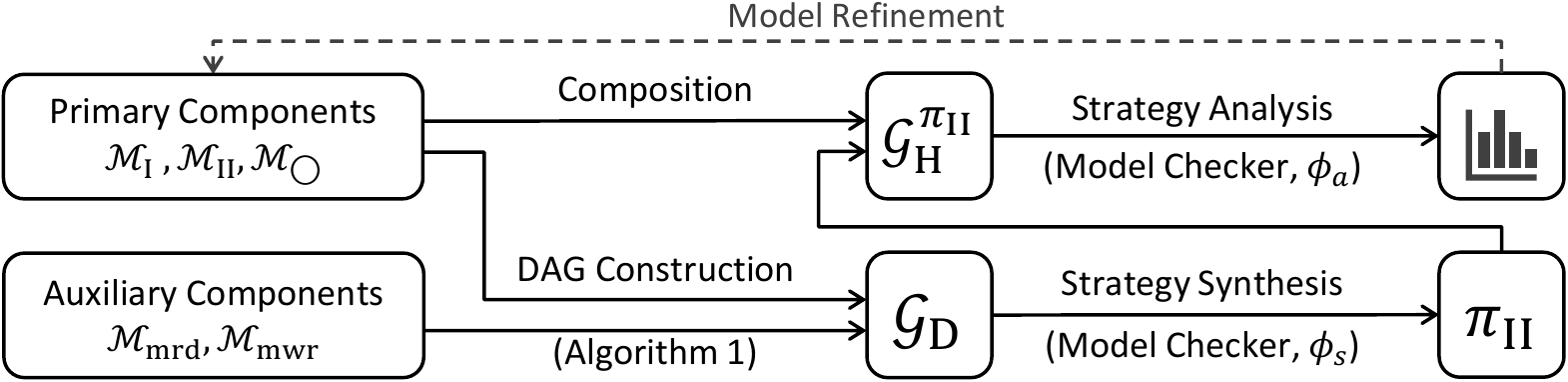}
	\caption{Synthesis and analysis framework based on the use of DAGs.} 
	\label{fig:framework}
\end{figure}

\begin{algorithm}[!t]
\caption{Procedure for DAG construction} \label{alg:dag_construction}
\DontPrintSemicolon
\SetKwRepeat{Do}{do}{while}
	\KwIn{Components $\M_{\Pone}, \M_{\Ptwo}, \M_{\Psto}, \M_{\Mwr}, \M_{\Mrd} $; initial state $\hat{s}_0$}
	\KwResult{DAG $\DAG$}
	\While{$\lnot(\mbox{end criterion})$}{
		\While(\CC{\PL{2} plays until a revealing attempt}){$a \neq \Reveal $}
		{
				$\M_{\Ptwo}.v_\B \leftarrow \Eff{a,v_\B} $,
				$\M_{\Mwr}.\mathit{write}(a,\texttt{++}\mathit{wr})$
		}
		\While(\CC{\PL{1} plays all delayed actions}){$\mathit{rd} \leqslant \mathit{wr}$}
		{
			$\M_{\Mrd}.\mathit{read}(a,\texttt{++}\mathit{rd})$,
			$\M_{\Pone}.v_\T \leftarrow \Eff{\beta(a),v_\T} $ \;
		}
		\uIf(\CC{\PLs~plays successful attempt}){draw $x \sim Brn(p(v_\T, v_\B))$}{
			$\M_{\Ptwo}.v_\B \leftarrow \M_{\Pone}.v_\T$,
			$ \mathit{wr} \leftarrow 0,\, \mathit{rd} \leftarrow 0$ \;
		}
		\lElse
		{ $ \mathit{rd} \leftarrow 0 $
		\CC{Unsuccessful attempt, forget \PL{1} actions}%
		}%
	}
\end{algorithm}

\algref{alg:strategy_synthesis} describes the procedure for strategy synthesis based on the DAG~$\DAG$,
 and an rPATL~\cite{chen2013automatic} synthesis query $\phi_\Syn$ that captures, for example, a safety requirement.
Starting with the initial location, the procedure checks whether $\query_\Syn$ is satisfied if action $\Reveal$ is performed at stage~$h$, and updates
 the set of feasible strategies $\Pi_{i}$ for subgame $\SUB_i$ until $h_{\max}$ is reached or $\phi_\Syn $ is not satisfied.%
\footnote{Failing to find a strategy at stage $i$ implies the same for all horizons of size $j>i$.}
Next, the set $\Pi_{i}$ is used to update the list of reachable end locations $\ell$ with new initial locations of reachable subgames that should be explored.
Finally, the composition of both $\HIG$ and $\Pi_\Ptwo^*$  resolves \PL{2} nondeterminism, where the resulting
model $\HIG^{\Pi_\Ptwo^*} $ is a Markov Decision Process (MDP) of complete information that can be easily used for further analysis.

\begin{algorithm}[tb]
\caption{Procedure for strategy synthesis} \label{alg:strategy_synthesis}
\DontPrintSemicolon
\SetKwRepeat{Do}{do}{while}
\KwIn{Initial location $(x_0,y_0)$, synthesis query $\query_\Syn$}
\KwOut{\PL{2} strategies $\Pi^*_{\Ptwo} $}
 $\ell \leftarrow \left[ (x_0,y_0) \right],\, i\leftarrow 0 $ \;
 \While(\CC{Explore all reachable subgames}){$i < | \ell | $}
 {
   $\hat{s}_{0} \leftarrow \left(\ell [i], \ell [i], \epsilon, 0, \Ptwo \right)$, $h \leftarrow 1, \Stop \leftarrow \bot $
   		\DD{Construct initial state}
   \While(\CC{Explore subgame till upper horizon}){$ h \leqslant h_{\max} \;\wedge\; \lnot\Stop$}{
     $\left(\pi_{\Ptwo},\varphi\right) \leftarrow 
     	\mathsf{Synth} \! \left(\SUB^{\pi_h}_{\hat{s}_0} ,\, \query_\Syn \right)$
     	\DD{Synthesize strategy for horizon h}
     \uIf(){$\pi_{\Ptwo} \neq \emptyset$}{
       $\Pi_{i} \leftarrow \Pi_{i} \cup \left(\pi_{\Ptwo},\pi_{h},\varphi \right)$, $h\texttt{++}$ \DD{Save synthesized strategy}
       }
     \lElse{
       $\Stop \leftarrow \top $
       }%
     }%
   $\mathsf{Prune} \left(\Pi_{t}\right) $, $\Pi^*_{\Ptwo} \leftarrow \Pi^*_{\Ptwo} \cup \Pi_t$
   		\DD{Prune subgame strategies}
   $\ell \leftarrow \ell \cdot \left(\mathsf{Reachable}\left(\Pi_{t}\right) \setminus \ell \right)$, $i\texttt{++}$
   	\DD{update reachability} 
   }
\end{algorithm}

%!TEX root = 00_CAV_2019_paper_78.tex
%%%%%%%%%%%%%%%%%%%%%%%%%%%%%%%%%%%%%%%%%%%%%%%%%%%%%%%%%%%%%%%%%%%%%%%%%%%%%%%
\section{Case Study}
\label{sec:casestudy}
%%%%%%%%%%%%%%%%%%%%%%%%%%%%%%%%%%%%%%%%%%%%%%%%%%%%%%%%%%%%%%%%%%%%%%%%%%%%%%%

In this section, we consider
a case study where
a human operator supervises a UAV prone to stealthy attacks on its
GPS sensor.
The UAV mission is to visit a number of targets 
 after being airborne from a known base (initial state), while avoiding hazard zones that are known a priori.
Moreover, the presence of adversarial stealthy attacks via GPS spoofing is assumed.
We use the DAG framework to synthesize strategies for both the UAV and 
 an operator advisory system (AS) that schedules geolocation tasks for the operator.

%==============================================================================
\subsubsection{Modeling.}

We model the system as a delayed-action game $\DAG$, where \PL{1} and \PL{2} represent the adversary
 and the UAV-AS coalition, respectively.
\figref{fig:system_components} shows the model primary and auxiliary components.
In the UAV model $\M_{\Uav}$, $x_\B \Eq (\V{x}_\B, \V{y}_\B)$ encodes the UAV
 belief, and
 $A_{\Uav} \Eq \{ \mathsf{N}, \mathsf{S}, \mathsf{E}, \mathsf{W},
  \mathsf{NE}, \allowbreak \mathsf{NW}, \mathsf{SE}, \mathsf{SW} \}$
 is the set of available movements.
The AS can trigger the action $\Activate$ to initiate a geolocation task, attempting to confirm the current location.
The adversary behavior is abstracted by $\M_{\Adv}$
 where $x_\T \Eq (\V{x}_\T, \V{y}_\T)$ encodes the UAV true location.
The adversarial actions are limited to
 one directional increment at most.\footnote{To detect aggressive attacks, techniques from literature
 			(\eg \cite{pajic_tcns17,pajic_csm17,jovanov_tac19}) can be used.}
If, for example, the UAV is heading $\mathsf{N}$, then the adversary set of actions
 is
 $\beta(\mathsf{N}) \Eq \{\mathsf{N},\mathsf{NE}, \mathsf{NW}\} $.
The auxiliary components
 $\M_{\Mwr}$ and $\M_{\Mrd}$ manage a FIFO memory stack
 {\larger[-1]$({m}_i)_{i=0}^{n-1} \in A_{\Uav}^n$}.
The last UAV movement is saved in ${m}_i$ by synchronizing $\M_{\Mwr}$ with  $\M_\Uav $ via $\Write$,
while $\M_{\Mrd}$ synchronizes with $\M_\Adv$ via $\Read$
 to read the next UAV action from ${m}_j$.
The subgame terminates whenever action $\Write$ is attempted and $\M_{\Mwr}$ is at state $n$ (\ie out~of~memory).

\begin{figure}[!t] %---------------------------
	\centering
	\includegraphics[scale=0.67]{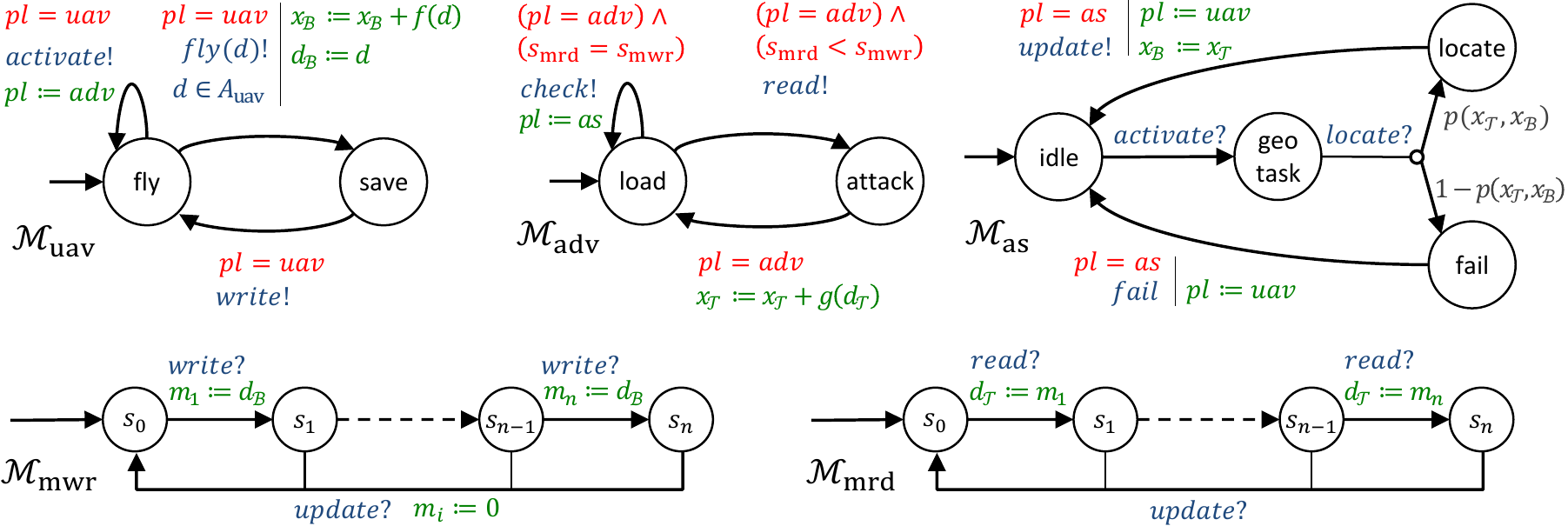}
	\caption{Primary DAG components: UAV ($\M_{\Uav}$),  adversary ($\M_{\Adv}$), and AS~($\M_{\Oas}$).
	Auxiliary DAG components:  memory write ($\M_{\Mwr}$) and memory read ($\M_{\Mrd}$) models,
	capturing the DAG representation. At  stage ${i}$, the next memory location to write/read~is~${m}_i$.} 
	\label{fig:system_components}
\end{figure}

The goal is to find strategies for the UAV-AS coalition
 based on the following:
\begin{itemize}
\item
\emph{Target reachability.}
To overcome cases where targets are unreachable due to hazard zones,
 the label $\Reach$ is assigned to the set of states with acceptable checkpoint locations (including the target) to render the objective incrementally feasible.
The objective for all encountered subgames is then formalized as
 $ \Probability_{\max} \left[ \eventually \; \Reach \right] \geqslant p_{\min} $ for some bound $p_{\min}$.
\item
\emph{Hazard Avoidance.}
Similar to target reachability, the label $\Hazard$ is assigned to states corresponding to hazard zones.
 The objective
 $\Probability_{\max} \left[ \always \; \lnot\Hazard \right] \geqslant p_{\min}$ is then specified for all encountered subgames.
\end{itemize}
By refining the aforementioned objectives, synthesis queries are used for both the subgames and the supergame. 
Specifically,~the~query
\begin{align}
\label{eqn:fi_sin}
 	%\DAG_i \colon\;
 	\query_\Syn(k) \!\coloneqq\! \llangle {\Uav} \rrangle \Probability_{\max=?} 
	\left[ \lnot\Hazard  \,\until^{\leqslant k}\,  ( \Locate \wedge \Reach )  \right]
	%\vspace{-4pt}
\end{align}
 is specified for each encountered subgame $\SUB_i$, where $\Locate$ indicates a successful geolocation task.
By following \algref{alg:strategy_synthesis} for a $q$ number of reachable subgames, the supergame is reduced to an MDP $\DAG^{ \{ \pi_i \}_{i=1}^{q}}$
(whose states are the reachable subgames), which is checked against the query
\begin{equation}
\label{eqn:fi_ana}
 	%\DAG^{ \{ \pi_i \}_{i=0}^{q}} \colon\; 
 	\query_\Ana (n) \!\coloneqq\! \llangle {\Adv} \rrangle \Probability_{\min,\max =?} 
	\left[ \eventually^{\leqslant n}\,  \Target  \right]
\end{equation}
 to find the bounds on the probability that the target is reached under a maximum number of geolocation tasks $n$.

%==============================================================================
\subsubsection{Experimental Results.}

\figref{fig:cs:setting} shows the map setting used for implementation.
The UAV's ability to actively detect an attack depends on both its belief and
 the ground truth.
Specifically,
the probability of success in a geolocation task mainly relies on the disparity between the belief and true locations,
 captured by 
 $\mathit{f}_{\mathrm{dis}} \colon \Eval{x_\B} \times \Eval{x_\T} \rightarrow [0,1]$,
obtained by assigning probabilities for each pair of locations
 		according to their features (\eg landmarks) and smoothed using a Gaussian 2D filter.
  		A thorough experimental analysis where probabilities are extracted from experiments with human operators is described in~\cite{elfar2019security}.
The set of hazard zones include the map boundaries to prevent the UAV from reaching boundary values.
Also, the adversary is prohibited from launching attacks for at least the first step,
a practical assumption to prevent the UAV model from infinitely bouncing around the target location.

\begin{figure}[!t]
	\centering
	\subfigure[Environment setup.]{%
		\label{fig:cs:setting}\includegraphics[height=1.6in]{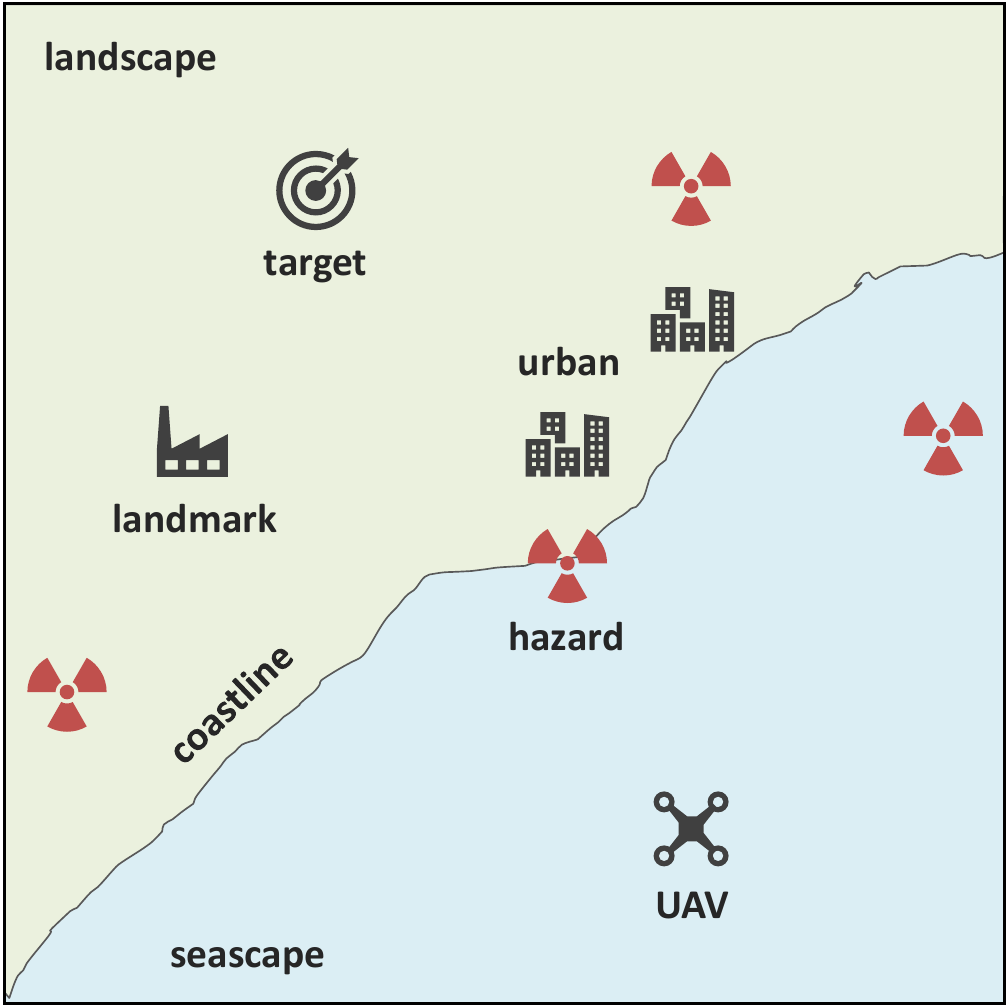}}%
	\hfill
	\subfigure[Supergame $\DAG$.]{%
		\label{fig:cs:supergame}\includegraphics[height=1.6in]{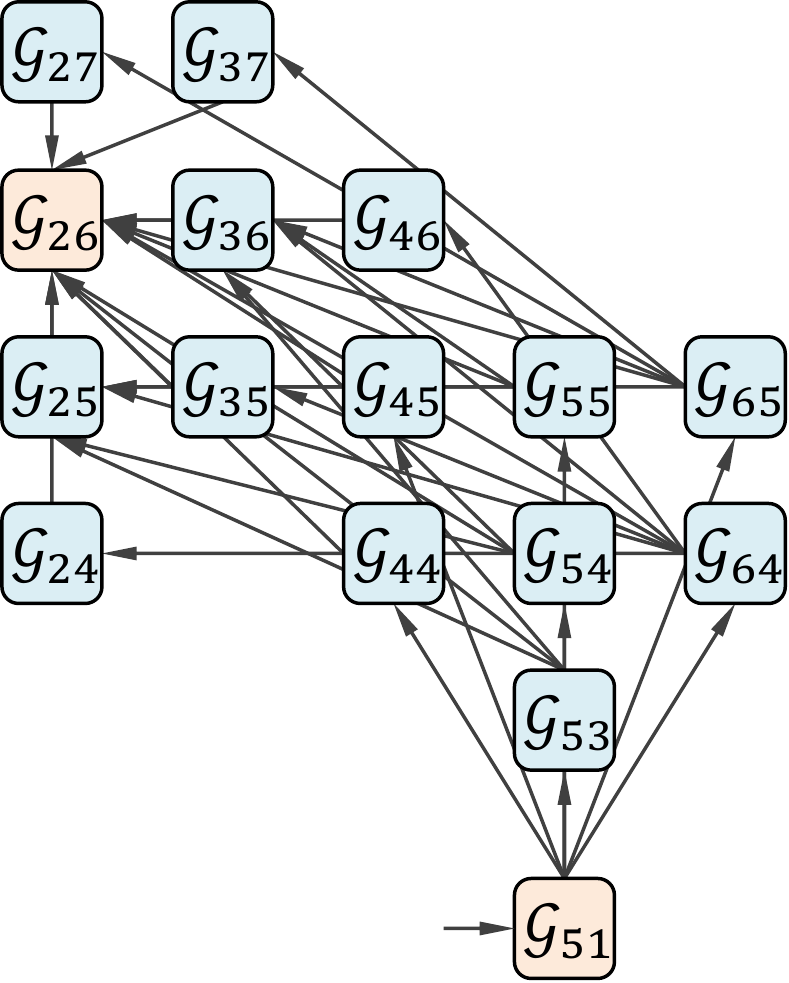}}%
	\hfill
	\subfigure[Protocols.]{%
		\label{fig:cs:protocols}\includegraphics[height=1.6in]{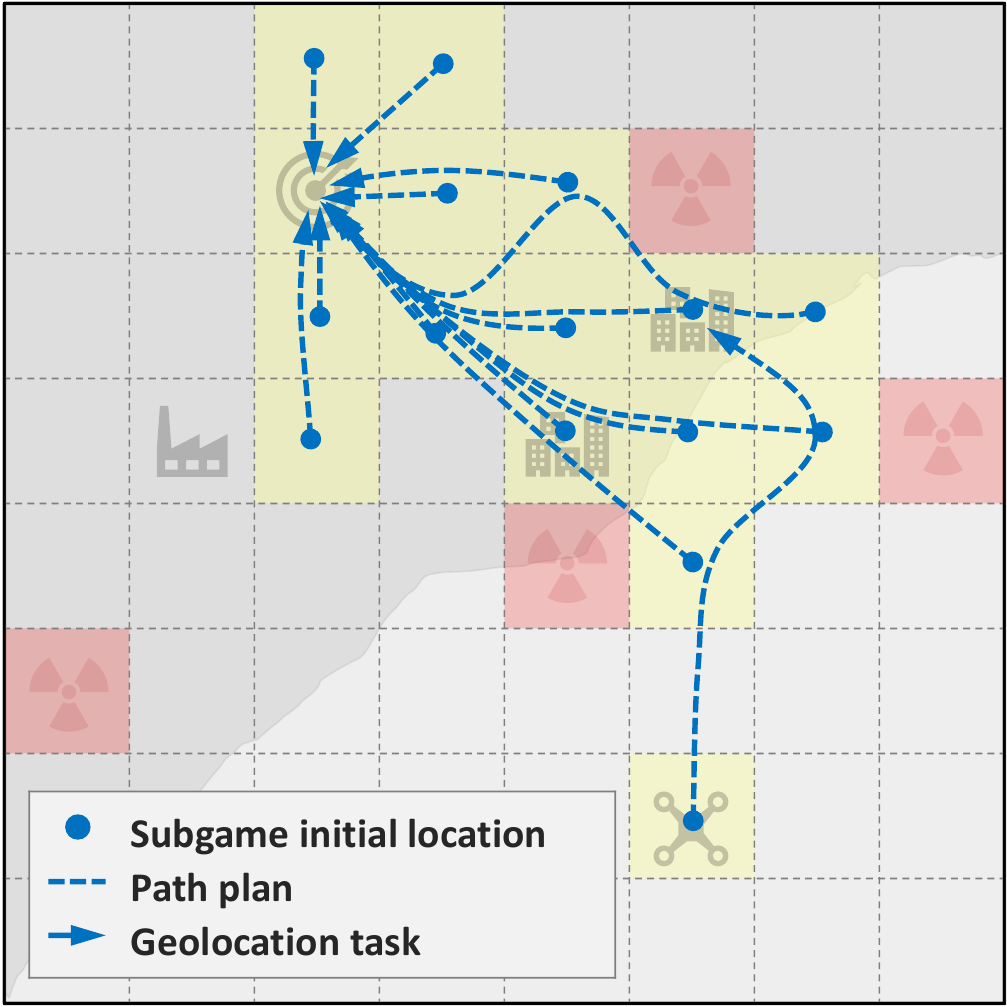}}%
	\caption{%
		{%		
		(a) The environment setup used for the case study;
		(b) the induced supergame MDP, where the subgames form its states; and
		(c) the synthesized protocols.
		}%		
		}
	\label{fig:cs}
\end{figure}

We implemented the model in \PRISM{}~\cite{chen2013prism,kwiatkowska2018prism}
and performed the experiments on an Intel Core i7 4.0 GHz CPU, with 10GB RAM dedicated to~the~tool.
\figref{fig:cs:supergame} shows the supergame obtained by following the procedure in \algref{alg:strategy_synthesis}.
A vertex $\SUB_{\V{xy}}$ represents a subgame (composed with its strategy) that starts at location $(\V{x},\V{y})$, while the outgoing
edges points to subgames reachable from the current one.
Note that each edge represents a probabilistic transition. Subgames with more than one outgoing transition imply nondeterminism
 that is resolved by the adversary actions. Hence, the directed graph depicts an MDP.

The synthesized strategy for $(h_\Adv\!\!=\!\!2,h\!\!=\!\!4)$ is demonstrated in \figref{fig:cs:protocols}.
For the initial subgame, \figref{fig:results:subgame} shows the maximum probability of a successful geolocation task if performed at stage $h$,
 and the remaining distance to target.
Assuming the adversary can launch attacks after stage
 $h_\Adv \!\!=\!\!2$, the detection probability is maximized by performing
 the geolocation task at step 4,
 and hazard areas can still be avoided up till $h\!\!=\!\!6$.
For $h_\Adv \!\!=\!\!1$, however, 
$h\!\!=\!\!3$ has the highest probability of success, which
 diminishes at $h\!\!=\!\!6$ as
 no possible flight plan exists without encountering a hazard zone.
The effect of the maximum number of geolocation tasks $(n)$
 on target reachability is studied by analyzing the supergame against $\query_\Ana$ as shown in \figref{fig:results:supergame}.
The minimum number of geolocation tasks to guarantee a non-zero 
 probability of reaching the target (regardless of the adversary strategy) is 3
 with probability bounds of $(33.7\%,94.4\%)$.

\newlength{\mywidth}
\newlength{\myheight}
\setlength{\mywidth}{2.2in}
\setlength{\myheight}{1.6in}
\begin{figure}[!t]
	\centering
	\subfigure{%
		\scalebox{0.8}{% This file was created by matlab2tikz.
\definecolor{cBlue}{rgb}{0.00000,0.44700,0.74100}%
\definecolor{cRed}{rgb}{0.85000,0.32500,0.09800}%
\definecolor{mycolor3}{rgb}{0.92900,0.69400,0.12500}%
\definecolor{mycolor4}{rgb}{0.49400,0.18400,0.55600}%
\definecolor{mycolor5}{rgb}{0.46600,0.67400,0.18800}%
\begin{tikzpicture}
%\node[text width=1cm] at (1.5,0.45) {(a)};
\pgfplotsset{width=6.5cm,compat=1.3}
\begin{axis}[%
%width=2.8in,
%height=1.2in,
%%% page width = 4.8in
width=\mywidth,
height=\myheight,
at={(0.8in,0.44in)},
scale only axis,
xmin=0,
xmax=7,
xtick={0, 1, 2, 3, 4, 5, 6, 7},
xlabel style={font=\normalsize}, %\color{white!15!black}},
xlabel={(a) Geolocation task at stage $h$},
ymin=0,
ymax=10,
ytick={0,2,4,6,8,10},
ylabel style={font=\normalsize , align=center},
ylabel={Distance to target},
axis background/.style={fill=white},
axis x line*=bottom,
axis y line*=right,
xmajorgrids,
ymajorgrids,
legend style={at={(0.97,0.97)}, anchor=north east, legend cell align=left, align=left, draw=white!15!black,font=\small}]

\addplot [smooth, color=black!40, solid, line width=2.0pt, mark size=1.6pt, mark=none, mark options={solid, fill=mycolor4, mycolor4}]
  table[row sep=crcr]{%
0 5.800\\
1 4.472\\
2 4.123\\
3 3.606\\
4 4.472\\
5 5\\
6 7.2\\
};
\addlegendentry{$h_{\Adv} \Eq 2$}

\end{axis}

\begin{axis}[%
%width=2.8in,
%height=1.2in,
%%% page width = 4.8in
width=\mywidth,
height=\myheight,
at={(0.8in,0.44in)},
scale only axis,
xmin=0,
xmax=7,
%tick style={very thin},
%xtick={0, 1, 2, 3, 4, 5, 6, 7},
xmajorticks=false,
ymin=0,
ymax=1,
ytick={0,0.2,0.4,0.6,0.8,1},
yticklabel style={/pgf/number format/fixed,
            /pgf/number format/precision=1,
            /pgf/number format/fixed zerofill},
y label style={font=\normalsize, align=center},
ylabel={Prob. of success $\query_\Syn$},
axis background/.style={fill=none},
axis x line*=bottom,
axis y line*=left,
%xmajorgrids,
%ymajorgrids,
legend style={at={(0.03,0.03)}, anchor=south west, legend cell align=left, align=left, draw=white!15!black,font=\small}]

\addplot [smooth, color=cRed, densely dotted, line width=0.5pt, mark size=1.6pt, mark=o, mark options={solid, fill=cRed, cRed}]
  table[row sep=crcr]{%
0 1.000\\
1 0.515\\
2 0.537\\
3 0.601\\
4 0.734\\
5 0.647\\
6 0.537\\
7 0.502\\
};
\addlegendentry{$h_{\Adv} \Eq 2$}

\addplot [smooth, color=cBlue, dashed, line width=0.5pt, mark size=1.6pt, mark=square, mark options={solid, fill=cBlue, cBlue}]
  table[row sep=crcr]{%
0 1.000\\ 1 0.515\\ 2 0.537\\ 3 0.601\\ 4 0.527\\ 5 0.502\\ 6 0.0\\ 7 0.0\\
};
\addlegendentry{$h_{\Adv} \Eq 1$}

\end{axis}

\end{tikzpicture}%
}
		\label{fig:results:subgame}}\hfill
	\subfigure{%
		\scalebox{0.8}{% This file was created by matlab2tikz.
\definecolor{cBlue}{rgb}{0.00000,0.44700,0.74100}%
\definecolor{cRed}{rgb}{0.85000,0.32500,0.09800}%
\definecolor{mycolor3}{rgb}{0.92900,0.69400,0.12500}%
\definecolor{mycolor4}{rgb}{0.49400,0.18400,0.55600}%
\definecolor{mycolor5}{rgb}{0.46600,0.67400,0.18800}%
\begin{tikzpicture}
%\node[text width=1cm] at (1.5,0.45) {(b)};
\pgfplotsset{width=7cm,compat=1.3}
\begin{axis}[%
width=\mywidth,
height=\myheight,
at={(0.8in,0.44in)},
scale only axis,
xmin=0,
xmax=10,
xtick={0, 1, 2, 3, 4, 5, 6, 7, 8, 9, 10},
xlabel style={font=\normalsize}, %\color{white!15!black}},
xlabel={(b) Max. no. of geolocation tasks $n$},
ymin=0,
ymax=1,
ytick={0.0,0.2,0.4,0.6,0.8,1.0},
yticklabel style={/pgf/number format/fixed,
            /pgf/number format/precision=1,
            /pgf/number format/fixed zerofill},
%tick style={blue!10},
ylabel style={font=\normalsize},
ylabel={Reachability bounds $\query_\Ana$},
axis background/.style={fill=white},
axis x line*=bottom,
axis y line*=left,
xmajorgrids,
ymajorgrids,
legend style={at={(0.97,0.16)}, anchor=south east, legend cell align=left, align=left, draw=white!15!black,font=\small},%
axis background style={fill=white}%
]

%==============================================================================

\addplot [name path=B, smooth, color=cBlue, solid, line width=0.5pt, mark size=1.8pt, mark=o,%
		  mark options={solid, fill=cBlue, cBlue}]
  table[row sep=crcr]{%
0.0   0.0                  \\
1.0   0.0                  \\
1.0   0.0                  \\
2.0   0.73899              \\
3.0   0.9436902300000001   \\
4.0   0.9885757436100001   \\
5.0   0.9977358742022701   \\
6.0   0.999555511013705    \\
7.0   0.999913079922302    \\
8.0   0.9999830305173686   \\
9.0   0.9999966893053536   \\
10.0  0.9999993542796102   \\
};
\addlegendentry{$\query_{\Ana,\max}$}

%==============================================================================

\addplot [name path=A, smooth, color=cRed, solid, line width=0.5pt, mark size=2.0pt, mark=triangle,%
 		  mark options={solid, fill=cRed, cRed, line width=0.5}]
  table[row sep=crcr]{%
0.0   0.0                  \\
1.0   0.0                 \\
2.0   0.0                 \\
2.0   0.0                 \\
3.0   0.33738677          \\
4.0   0.6376609953        \\
5.0   0.82068015994097    \\
6.0   0.9161146671467636  \\
7.0   0.9620444284425603  \\
8.0   0.983170325597471   \\
9.0   0.9926301449895001  \\
10.0  0.9967975760207504  \\
};
\addlegendentry{$\query_{\Ana,\min}$}

%==============================================================================

\addplot[fill=gray,fill opacity=0.3, forget plot] fill between[of=A and B];

%==============================================================================

\addplot [name path=C, smooth, color=cBlue, densely dashed, line width=0.5pt, mark size=1.0pt, mark=*,%
		  mark options={solid, fill=cBlue, cBlue, line width=0.5}]
  table[row sep=crcr]{%
2.0    0.73899                 \\
3.0    0.20470023000000004     \\
4.0    0.04488551361000004     \\
5.0    0.009160130592269944    \\
6.0    0.0018196368114349326   \\
7.0    3.5756890859706125E-4   \\
8.0    6.995059506653689E-5    \\
9.0    1.3658787984982546E-5   \\
10.0   2.6649742566053902E-6   \\
};
\addlegendentry{$\Updelta \query_{\Ana,\min}$}

%==============================================================================

\addplot [name path=C, smooth, color=cRed, densely dashed, line width=0.5pt, mark size=1.8pt, mark=x,%
		  mark options={solid, fill=cRed, cRed, line width=0.5}]
  table[row sep=crcr]{%
3.0    0.33738677             \\
4.0    0.3002742253           \\
5.0    0.18301916464097       \\
6.0    0.0954345072057936     \\
7.0    0.045929761295796734   \\
8.0    0.021125897154910644   \\
9.0    0.009459819392029112   \\
10.0   0.004167431031250279   \\
11.0   0.001817572087053465   \\
};
\addlegendentry{$\Updelta \query_{\Ana,\max}$}

\end{axis}

\end{tikzpicture}%
}
		\label{fig:results:supergame}}
	\caption{
		Analysis results for (a) subgame $\SUB_{51}$ and (b) supergame $\DAG$.}
	\label{fig:results}
\end{figure}

The experimental data obtained for this case study are listed
 in~\tableref{tab:experimental_results}.
%%%Note that, 
For the same grid size, more complex maps require more time
 for  synthesis  while the state space size remains unaffected.
The state space grows exponentially with the explored horizon size,
 \ie $ \mathcal{O} \left( (|A_{\Uav}||A_{\Adv}| )^{h} \right) $, 
 and is typically slowed by, \eg the presence of hazard areas, since the branches
 of the game transitions are trimmed upon encountering such areas.
Interestingly, for $h\!\!=\!\!6$ and $h\!\!=\!\!7$, while the model construction time (size) for $h_\Adv \!\!=\!\!1$
 is almost twice (quadruple) as those for $h_\Adv \!\!=\!\!2$,
the time for checking $\query_\Syn$ declines in comparison. This reflects the fact that, in case of $h_\Adv \!\!=\!\!1$ compared to $h_\Adv \!\!=\!\!2$,
the UAV has higher chances to reach a hazard zone for the same $k$, leading to a shorter time for model checking.

\begin{table}[!tb]
\caption{
	Results for strategy synthesis using queries $\query_\Syn$ and $\query_\Ana$.}
\label{tab:experimental_results}
\setlength{\tabcolsep}{10pt}%
\newcommand{\None}{\multicolumn{1}{c}{--}}%
\resizebox{\columnwidth}{!}{%
\begin{tabular}{lll|rrr|rrr}
	\hline	
	\multicolumn{3}{c}{Subgame {\larger[-2]$\SUB_{51}$}} &
	\multicolumn{3}{|c}{Model Size} &
	\multicolumn{3}{|c}{Time (sec)} \\
	\hline
	Map & $t_{adv}$ & $k$ &
	\multicolumn{1}{|c}{States} &
	\multicolumn{1}{c}{Transitions} &
	\multicolumn{1}{c}{Choices} &
	\multicolumn{1}{|c}{Model} &
	\multicolumn{1}{c}{$\query_\Syn$} &
	\multicolumn{1}{c}{$\query_\Ana$} \\
	\hline
	%\rule{0pt}{3ex}
	$8 \times 8$ 	
		& 1		
		 & 4 & \num{11608} & \num{17397} & \num{15950} &  	  {2.810} & {0.072} & \None \\
		&& 5 & \num{57129} & \num{87865} & \num{83267} &  	  {14.729} & {0.602} & \None \\
		&& 6 & \num{236714} & \num{366749} & \num{359234} &   {62.582} & {1.293} & \None \\
		&& 7 & \num{876550} & \num{1365478} & \num{1355932} & {231.741} & {6.021} & \None \\
	\cline{2-9}
		& 2
		 & 4 & \num{6678} & \num{9230} & \num{8394} &  	      {2.381} & {0.042} & \None \\
		&& 5 & \num{33904} & \num{48545} & \num{45354} &  	  {10.251} & {0.367} & \None \\
		&& 6 & \num{141622} & \num{204551} & \num{198640} &   {37.192} & {1.839} & \None \\
		&& 7 & \num{524942} & \num{763144} & \num{754984} &   {145.407} & {8.850} & \None \\
	\hline
	 \multicolumn{3}{c|}{Supergame {\larger[-2]$\DAG$}}
	 	& \num{6212} & \num{8306} & \num{6660} & {2.216} & \None & {2.490} \\
	\hline
\end{tabular}%
}%
\end{table}

%!TEX root = 00_CAV_2019_paper_78.tex
%%%%%%%%%%%%%%%%%%%%%%%%%%%%%%%%%%%%%%%%%%%%%%%%%%%%%%%%%%%%%%%%%%%%%%%%%%%%%%%
\section{Discussion and Conclusion}
\label{sec:conclusion}
%%%%%%%%%%%%%%%%%%%%%%%%%%%%%%%%%%%%%%%%%%%%%%%%%%%%%%%%%%%%%%%%%%%%%%%%%%%%%%%

In this paper, we introduced DAGs and showed how they can simulate HIGs by delaying players' actions.
We also derived a DAG-based framework for strategy synthesis and analysis using off-the-shelf SMG model checkers.
Under some practical assumptions, we showed that DAGs can be decomposed into independent subgames,
utilizing parallel computation to reduce the time needed for model analysis, as well as the size of the state space.
We further demonstrated the applicability of the proposed framework on a case study focused on synthesis and analysis of active attack detection strategies for UAVs prone to cyber attacks. 

DAGs come at the cost of increasing the total state space size as $\M_\Mrd$ and $\M_\Mwr$ are introduced.
This does not present a significant limitation due to the compositional approach towards strategy synthesis using subgames. 
However, the synthesis is still limited to model sizes that off-the-shelf tools can handle.

The concept of delaying actions implicitly assumes that the adversary knows the UAV actions a priori.
This does not present a concern in the presented case study as an abstract (\ie nondeterministic) adversary model is analogous to synthesizing against the worst-case attacking scenario.
Nevertheless,
strategies synthesized using DAGs (and SMGs in general) are inherently conservative.
Depending on the considered system, this can easily lead to no feasible solution.

The proposed synthesis framework ensures preservation of safety properties. Yet, general
reward-based strategy synthesis is to be approached with care.
For example, rewards dependent on the belief can appear in any state, and exploring hypothetical branches is not required.
However, rewards dependent on a state's true value should only appear in proper states, and all hypothetical branches are to be explored.
A detailed investigation of how various properties are preserved by DAGs, along with multi-objective synthesis, is a direction for future~work.

%%% BIBLIO %%%%%%%%%%%%%%%%%%%%%%%%%%%%%%%%%%%%%%%%%%%%%%%%%%%%%%%%%%%%%%%%%%%%
\label{page:before_biblio}
\bibliographystyle{splncs04}
\bibliography{Biblio}

\end{document}